\documentclass[11pt]{article} 
\usepackage{jheppub}
\usepackage{bvrshorthand}
\usepackage{xcolor}

\newcommand{\abs}[1]{\left\lvert#1\right\rvert}
\newcommand{\hb}{{\bar h}}
\newcommand{\Hb}{{\bar H}}

\newtheorem{theorem}{Theorem}[section]
\newtheorem{lemma}[theorem]{Lemma}
\newtheorem{corollary}[theorem]{Corollary}
\newtheorem{proposition}[theorem]{Proposition}

\theoremstyle{remark}
\newtheorem{remark}[theorem]{Remark}

\definecolor{amethyst}{rgb}{0.6, 0.4, 0.8}

\renewcommand{\geq}{\geqslant}
\renewcommand{\leq}{\leqslant}

\title{Universality of the microcanonical entropy at large spin}

\author{Sridip Pal$^{\tau}$, Jiaxin Qiao$^{\tau+1}$, Balt C. van Rees$^{-1/\tau}$}
\affiliation{$^{\tau}$Walter Burke Institute for Theoretical Physics,  California Institute of Technology,  Pasadena, CA, USA  \\  
	$^{\tau+1}$ Laboratory for Theoretical Fundamental Physics, Institute of Physics, École Polytechnique Fédérale de Lausanne (EPFL), CH-1015 Lausanne, Switzerland\\
	$^{-1/\tau}$CPHT, CNRS, Ecole Polytechnique, Institut Polytechnique de Paris, Palaiseau, France }

\abstract{We consider rigorous consequences of modular invariance for two-dimensional unitary non-rational CFTs with $c > 1$. Simple estimates for the torus partition function can lead to remarkably strong results. We show in particular that the spectral density of spin-$J$ operators must grow like $\exp\left( \pi \sqrt{\frac{2}{3}(c-1) J} \right)/\sqrt{2J}$ in any twist interval at or above $(c-1)/12$, with a known twist-dependent prefactor. This proves that the large $J$ spectrum becomes dense even without averaging over spins. For twists below $(c-1)/12$ we establish that the growth must be strictly slower. Finally, we estimate how fast the maximal gap between two spin-$J$ operators must go to zero as $J$ becomes large.}

\begin{document}
\maketitle
\section{Introduction}
In this work we will be concerned with two-dimensional unitary ``non-rational'' CFTs. We assume a modular invariant torus partition function, central charge $c > 1$, a unique normalizable vacuum, and a twist gap in the spectrum of non-identity Virasoro primaries.\footnote{We note that an explicit construction of such theories is a difficult problem, in contrast to the abundant rational conformal field theories. Recent possible constructions of non-rational CFTs can be found in \cite{Antunes:2022vtb,Antunes:2024mfb}.} In such theories we consider the large-spin asymptotic behavior of the density
\begin{equation}
	\rho_J(\Delta)
\end{equation}
of local operators with scaling dimension $\Delta$ and spin $J$. A naive use of modular invariance dictates that
\begin{equation}
	\label{rhoJnaiveintro}
	2 \rho_J(J + 2h) \underset{J \to \infty}{\leadsto} \frac{1}{\sqrt{2J}} \exp\p{\pi \sqrt{\frac{2}{3}(c-1) J}} \rho_{\rm c}(h)\,,
\end{equation}
essentially due to the Virasoro identity block in the dual channel. Here $\rho_{\rm c}(h)$ is a continuous function supported at $h \geq (c-1)/24$, defined in equation \eqref{def:rho0} below, and $h$ is half the twist $\Delta - J$. Our main objective in this paper is to translate this well-known but mathematically imprecise statement into rigorous claims that provide quantitative predictions for the large-spin spectrum.

The main issue with equation \eqref{rhoJnaiveintro} has to do with averaging. For one, the large $J$ limit of
\begin{equation}
	\sqrt{2J} \exp\p{-\pi \sqrt{\frac{2}{3}(c-1) J}} \rho_J(J + 2h)
\end{equation}
is not at all expected to exist, at least not pointwise in $h$. Indeed, this would be strange: $\rho_{\rm c}(h)$ is continuous whereas $\rho_J(\Delta)$ is typically a sum of delta functions. Furthermore, rigorous statements in the present context generally follow from so-called \emph{Tauberian} theorems. Those, however, typically provide a result about an ``averaged'' large $J$ limit, that is only after summing both sides from $J = 0$ until some large $J_\text{max}$. So the precise question to be asked is \emph{what kind of averaging is required, both in $J$ and in $h$, to turn equation \eqref{rhoJnaiveintro} into a rigorous statement?}

Our answer is provided in theorem \ref{theorem:error} below. In short, averaging over $J$ is \emph{not} required, and if we average over $h$ with a smooth test function $\varphi(h)$ then the averaging window can shrink as fast as $J^{-1/4 + \epsilon}$ for any $\epsilon > 0$. We leave a precise statement to the main text and move on to discuss some of its consequences.

\subsection*{Corollaries}

First, our theorem implies that the spectrum becomes \emph{dense}: every twist interval overlapping with the support of $\rho_{\rm c}(h)$ will for large enough $J$ contain an arbitrarily large number of Virasoro primaries. In fact, we show that this result holds for \emph{every} sufficiently large spin $J$: individual spins cannot misbehave.\footnote{This excludes the unlikely but heretofore not disproven scenario where the spectrum for individual spins remains integer spaced but with offsets such that the spin-averaged spectrum becomes dense.}

Second, corollary \ref{corr:count} uses theorem \ref{theorem:error} to provide a precise result for the average operator density. Taking the logarithm of the stated result, we prove that the spin-$J$ microcanonical entropy $S_J(h)$ at twist $\Delta - J = 2h$ becomes \emph{universal} at large $J$. In equations:
\begin{align}\label{eq:entropy}
	h &\geq \frac{c-1}{24}\,: & &\,\,S_J(h) = \pi \sqrt{\frac{2}{3}(c-1)J} - \frac{1}{2}\log(2J) + S_{\rm c}(h) + o(J^0)\\
	h &< \frac{c-1}{24}\,: & \lim_{J \to \infty} &\,\,S_J(h) - \pi \sqrt{\frac{2}{3}(c-1)J} + \frac{1}{2}\log(2J)  = - \infty\nn
\end{align}
The last equation may appear strange, but essentially says that the entropy growth has to be strictly subleading to the shown behavior. The constant term $S_{\rm c}(h)$ corresponds to the microcanonical entropy associated to $\rho_{\rm c}(h)$. Here the microcanonical entropy can be defined as usual, \emph{i.e.}~as the density of spin-$J$ operators in a finite fixed interval centered around $h$. However, as in theorem \ref{theorem:error}, the above result holds even if one defines the microcanonical entropy by counting states in an interval whose size shrinks as $J^{-1/4 + \epsilon}$, for any $\epsilon > 0$. This strengthens and generalizes to all twists a similar result obtained earlier by two of us  \cite{Pal:2023cgk} for the large-spin operators with $h$ centered around $(c-1)/24$.\footnote{The related claim that $T_\text{gap} \leq A$ was discussed earlier in \cite{Collier:2016cls}, where it was credited to T. Hartman, and \cite{Afkhami-Jeddi:2017idc,Benjamin:2019stq}. }

Our third result concerns the \emph{maximal} spacing in the spectrum. Consider any finite open interval in $h$, supported in $[(c-1)/24, \, \infty)$. Of course the maximal spacing between operators in this interval must go to zero at large $J$, since a spectrum becoming dense implies that every fixed subinterval must eventually contain at least one operator. Corollary \ref{corr:spacing} provides a more quantitative estimate: the maximal spacing is given by
\begin{equation}
	o(J^{-1/4 + \epsilon})\,
\end{equation}
for any $\epsilon > 0$.\footnote{The astute reader will have noticed that this is not an immediate corollary of the previous paragraph. Our demonstration relies on local uniformity in $h$ that is proven in our main theorem \ref{theorem:error}.} Given the much more rapid growth of the average density, a power-law vanishing rate for the maximal spacing appears rather slow. It would be therefore be interesting to see if it can nevertheless be saturated, or if perhaps our estimate can be improved.

Finally we note that our bounds follow directly from simple estimates on the vacuum and non-vacuum terms in the modular invariant partition functions and do not make use of existing Tauberian theorems.

\subsection*{Context}
Modular invariance dictates the leading behavior of the torus partition function $Z(\beta)$ as $\beta\to 0$. Following Cardy \cite{Cardy:1986ie}, we can apply an inverse Laplace transform to the vacuum term in the dual channel to obtain an estimate of sorts for the large $\Delta$ limit of the density of states. A large $J$ estimate follows instead from considering the two-variable torus partition function $Z(\beta_L,\beta_R)$ as $\beta_R \to 0$. The naive equation \eqref{rhoJnaiveintro} is then obtained by applying \emph{two} inverse Laplace transforms to the leading term.

There are rigorous so-called \emph{Tauberian} theorems that discuss under what kind of averaging these operations are valid, see for example the textbook \cite{korevaar2004tauberian}. Tauberian theorems entered the field of the conformal bootstrap in 2012, when they were used in \cite{Pappadopulo:2012jk} to study the large $\Delta$ limit of the OPE density of CFT four-point functions in general dimensions.

Indeed, there is a strong analogy between the modular \cite{Hellerman:2009bu} and four-point function conformal bootstrap. In both cases we can expect to gain information about large $\Delta$ behavior by studying a diagonal limit, and about large $J$ behavior by studying an off-diagonal or lightcone limit. In fact, this qualitative similarity can be made precise in certain scenarios or kinematical regimes.
%see \cite{Zamolodchikov:1987avt,Maldacena:2015iua,Maloney:2016kee,Das:2017cnv,Hartman:2019pcd}.
For example,  the sphere four-point function can be analyzed by mapping it to a pillow $\mathbb{T}^2/\mathbb{Z}_2$, where crossing symmetry manifests itself as modular covariance \cite{Zamolodchikov:1987avt,Maldacena:2015iua}. Such a mapping was leveraged to prove the absence of a bulk point singularity \cite{Maldacena:2015iua}, to investigate the possibility of reconstructing a correlator given the contribution from the ``light'' spectrum \cite{Maloney:2016kee}, and to find the large $\Delta$ behavior (light-light-heavy) of OPE coefficients \cite{Das:2017cnv}. A similar mapping was used in \cite{Hartman:2019pcd} to write the torus partition function of a CFT $\mathcal{A}$ as a four-point function of $\mathbb{Z}_2$-twist operators in $(\mathcal{A}\times\mathcal{A})/\mathbb{Z}_2$; as a result, modular invariance of the original partition function became crossing symmetry of the four-point function.  This analogy is explained in \cite{Pal:2022vqc} in the context of lightcone bootstrap. 

The diagonal limit for four-point functions was analyzed further in \cite{Qiao:2017xif,Mukhametzhanov:2018zja}, and for torus partition functions in \cite{Mukhametzhanov:2019pzy}. It is particularly interesting to contrast our results with prior bounds on the maximal gap between two Virasoro primary operators. For unitary 2D CFTs with $c>1$, the spectrum of Virasoro primaries $\{\Delta_n\}$ satisfies $\limsup \Delta_{n+1}-\Delta_n \leqslant 1$, as conjectured in \cite{Mukhametzhanov:2019pzy} and proven in \cite{Ganguly:2019ksp,Mukhametzhanov:2020swe}. The authors of \cite{Mukhametzhanov:2020swe} further refined the result with respect to spin: in unitary 2D CFTs with $c>1$, the spectrum of Virasoro primaries with spin $J$, $\{\Delta_{n,J}\}$, satisfies $\limsup \Delta_{n+1,J}-\Delta_{n,J} \leqslant 2$. Note that these works assumed only a gap on the spectrum of scaling dimensions of Virasoro primaries, \emph{i.e.} $\Delta_1>0$. Our work provides a stronger claim, but only if we further assume a twist gap. Indeed, translating our results from $J$ to $\Delta$ shows that we have proven that the level spacing asymptotes to zero at large scaling dimension, $\limsup(\Delta_{n+1}-\Delta_n)=0$, provided the twist $\Delta - J$ is contained in some bounded interval above $(c-1)/12$. Furthermore, the decay rate is faster than $\Delta^{-1/4+\epsilon}$ for any $\epsilon>0$.

A rigorous analysis of the off-diagonal or lightcone limit proved to be more stubborn, even though intuitive arguments also date back to 2012 \cite{Fitzpatrick:2012yx,Komargodski:2012ek}. A first theorem for four-point functions was provided two years ago in \cite{Pal:2022vqc}, but only for the large $J$ behavior of the \emph{leading} double-twist Regge trajectory. The authors of \cite{Pal:2022vqc} also analyzed torus partition functions and proved the existence of the predicted \cite{Collier:2016cls} infinite number of large-spin operators whose twists $\Delta - J$ accumulate at $(c-1)/12$. As we mentioned above, some further theorems on these ``leading twist'' operators were proved by two of us in \cite{Pal:2023cgk}.

Proving theorems for the lightcone bootstrap at higher twists is more difficult because of the need to perform \emph{two} inverse Laplace transforms, which naively invalidates the direct use of Tauberian theorems. Last year one of us was able to solve this problem \cite{vanRees:2024xkb} using Vitali's theorem in complex analysis. This led to a rigorous existence proof of \emph{all} the double-twist Regge trajectories in a general unitary CFT four-point function with a twist gap.

The application the methods of \cite{vanRees:2024xkb} to the modular bootstrap problem was the original motivation for this work. It turned out that we were able to obtain significantly stronger results, using neither Tauberian theorems nor Vitali's theorem. The reason we can go beyond standard Tauberian theorems is because we have more control over the partition function than just its $\beta_R \to 0$ limit. A crucial difference appears in the proof of lemmas \ref{lemma:Fnonvac1} and \ref{lemma:Fnonvac2}, where we use modular invariance \emph{twice} to bound the non-universal term. This procedure is reminiscent of the use of the bounds obtained in \cite{Hartman:2014oaa} to obtain results on the non-summed operator density at large $\Delta$ in \cite{Mukhametzhanov:2019pzy}.

The application of the ideas presented in this work to the four-point function bootstrap will appear elsewhere.

\subsection*{Connection to chaos and thermalization}
A natural physics context for our work emerges in the realms of chaos and thermalization. In many quantum systems, the eigenvalues of the Hamiltonian in an appropriate regime such as large energy exhibit statistical features. A standard way to define such statistics is to rescale the energy eigenvalues with the mean density, and to ask whether the distribution of energy eigenvalues in the rescaled variables has statistical features. In integrable theories, the distribution is Poissonian \cite{Berry:1977levelclustering} while in chaotic theories, the distribution mimics random matrix theory and exhibits features of chaos such as level repulsion \cite{Bohigas:1983er}.
Most famously the energy levels of atomic nuclei can be modeled by random matrix theory \cite{Wigner:1951stat,PorterThomas:1956,Mehta:1960stat,Mehta:1967book}. Another notable model is the distribution of the non-trivial zeros of the Riemann zeta function with large imaginary part \cite{montgomery1973pair,odlyzko1987distribution}. Yet another set of examples come from hyperbolic manifolds \cite{marklof2006arithmetic}, which has a close connection to the conformal bootstrap \cite{Kravchuk:2021akc}. 

In recent years, there has been growing amount of evidence of the relevance of chaotic CFTs in theoretical physics, in particular in the context of low dimensional holography and black hole physics \cite{Cotler:2016fpe} via AdS/CFT dualities. For example, Jackiw–Teitelboim (JT) gravity is dual to double-scaled random matrix theory~\cite{Saad:2019lba}. Beyond the realm of AdS/CFT we expect the high energy eigenstates to behave thermally in an ergodic quantum system, a phenomenon known as \textit{eigenstate thermalization hypothesis} a.k.a ETH \cite{Srednicki:1994mfb}. 

In CFTs, we can make progress in understanding chaos and thermalization with various physical assumptions. Our ability to do so is tied with the fact that CFT observables exhibit universality at large quantum numbers (say energy or spin). For example, evidence in support the ETH in 2D CFTs can be found in \cite{Kraus:2016nwo,Das:2017vej,Basu:2017kzo, Brehm:2018ipf,Romero-Bermudez:2018dim,Hikida:2018khg,Collier:2019weq,Chen:2024lji}. The signature of chaos in generic 2D CFTs has been studied via the butterfly effect \cite{Roberts:2014ifa}. To extract the features of chaos, effective field theories have been developed \cite{Haehl:2018izb,Anous:2020vtw,Altland:2020ccq,Choi:2023mab}. In recent years, harmonic analysis has proven to be a useful tool to characterize chaos, see \cite{Boruch:2025ilr,Haehl:2023mhf,Haehl:2023xys,DiUbaldo:2023qli, Collier:2021rsn,Haehl:2023tkr}, which is built upon \cite{Benjamin:2021ygh}. However, these result mostly rely on existence of chaotic CFTs. In particular, it is expected that the Cardy formula for the average density, of states is true even when averaged over a very tiny window in energy. Subsequently, the correction to this mean density is expected to exhibit chaotic features and has been studied through the lens of harmonic analysis in \cite{Boruch:2025ilr,Haehl:2023mhf,Haehl:2023xys,DiUbaldo:2023qli, Haehl:2023tkr}. This begs the question for what kind of CFTs these expectations are a reality.

As we discussed, to be able to define statistics of energy eigenstates and probe chaos, we should study the energy states at the scale of mean level spacing. Given the density of states in a 2d CFT has a Cardy like growth, a bare minimum ``green light'' to have a notion of statistics and hence chaos is to have a spectrum such that the maximum level spacing goes to $0$ for states with large quantum number. Furthermore, concepts like ETH rely on a single eigenstate behaving thermally. While there is a lot of evidence for ETH in 2D CFTs, most of these results are true when the microcanonical window has an order one size as opposed to having a single energy eigenstate. Thus a true probe of ETH would be when we can shrink the size of the microcanonical window much smaller than order one, possibly comparable to the mean level spacing and still have a universal statement about thermality.

The theorem that we prove in this paper shows that indeed the spacing goes to $0$ at large spin in non-rational CFTs. We see this as a minimal necessary condition for chaos. This provides with a reasonable hope that such theories are indeed chaotic. 

\subsection*{Overview}
Our paper is structured as follows. Section \ref{sec:setup} provides the exact setup and axioms. In section \ref{sec:laplacelargespin} we first explain that it suffices to estimate the Laplace transform 
\begin{equation}
	F(\beta_L,J) \colonequals 2 \int_{T_\text{gap}}^\infty dh\, \rho_J(J + 2h) e^{-\beta_L (h - A)}
\end{equation}
in the right half plane with $\beta_L > 0$. We then proceed to provide the necessary bounds and estimates, both for the universal term (corresponding to the vacuum in the dual channel) and the non-universal term. We then put everything together in section \ref{sec:theorem}, where we state and prove our main theorem and the two corollaries mentioned above. A brief outlook concludes the paper.

\section{Setup}
\label{sec:setup}
The object of our investigation will be the spin-$J$ spectral densities
\begin{equation}
    \rho_J(\Delta)\,
\end{equation}
of non-identity Virasoro primaries in a two-dimensional unitary CFT. For each integer\footnote{Note that $J$ is allowed to be negative. In a parity-symmetric theory $\rho_J(\Delta) = \rho_{-J}(\Delta)$, but we will not assume this. In this work we will always take the large positive $J$ limit, but our results hold equally in the large negative $J$ limit.} spin $J \in \Z$ the spectral density is positive and integrable over $\Delta \in \R$ in the Riemann-Stieltjes sense. Unitarity dictates that its support is limited to $\Delta \geq |J|$, but in this paper we further impose a \emph{twist gap} $2T_\text{gap} > 0$ such that for all $J$
\begin{equation}
	\supp\p{\rho_J(\Delta)} \subset \left\{\Delta \in \R \, : \, \Delta - |J| > 2T_\text{gap}\right\}\,.
\end{equation}
We will suppose that the growth at large $|J|$ and $\Delta$ is sufficiently benign such that
\begin{equation}
	\sum_J \int d\Delta \, \rho_J(\Delta) e^{- \beta \Delta - \mu J}
\end{equation}
is absolutely convergent as long as $\Re(\beta) > |\Re(\mu)| \geq 0$. 

To fix ideas consider the case where the theory has a discrete spectrum. Then we have
\begin{equation}
	\rho_J(\Delta) = \sum_k \delta(\Delta - \Delta_k^{(J)})\,,
\end{equation}
where each energy (or, more accurately, scaling dimension) $\Delta_k^{(J)}$ with $k \in \{1,2,3,\ldots\}$ lies at or above $|J| + 2T_\text{gap}$. Finiteness of the above sum-plus-integral implies further conditions: for example,  that there are only finitely many energies below any given threshold $\Delta_\text{max}$. However, in this paper we will not necessarily assume a discrete spectrum.

\subsubsection*{Two-variable spectral density}
It will be convenient to introduce $h = (\Delta- J)/2$ and $\hb = (\Delta + J)/2$ and to introduce the two-variable density
\begin{equation}\label{rhohhbdefn}
	\rho(h,\hb) \colonequals 2\sum_{J} \rho_J(h + \hb)\delta(J - \hb + h)\,,
\end{equation}
where the factor of 2 ensures that $\rho(h,\hb) = \sum_k \delta(h - h_k) \delta(\hb - \hb_k)$ in the case of a discrete spectrum. This density is supported in the region $h,\hb > T_\text{gap}$,  and such that
\begin{equation}
	\int_{T_\text{gap}}^\infty dh \int_{T_\text{gap}}^\infty d\hb \, \rho(h,\hb)e^{-\b_L h - \b_R \hb}
\end{equation}
is finite as long as the complex numbers $\b_L, \b_R$ have positive real parts.

Below we will also consider integrals of the form
\begin{equation}\label{finiteHbIntegral}
	\int^{\Hb}_{T_\text{gap}}  d\hb\,\rho(h,\hb)\,.
\end{equation}
By convention, the upper limit of the integral is understood as $\lim_{\epsilon \searrow 0} \int^{\Hb + \epsilon}(\ldots)$ which implies continuity from the right in $\Hb$. For each finite $\Hb$ this produces a positive density over $h$ which is again integrable in the Riemann-Stieltjes sense. An integral of the form $\int^{b}_{a}  d\hb$ is then understood as $\int^{b}_{T_\text{gap}}  d\hb-\int^{a}_{T_\text{gap}}  d\hb$.

\subsubsection*{Partition function}
Our goal will be to obtain universal constraints on $\rho_J(\Delta)$ at large $J$. We will do so using modular invariance of the CFT torus partition function, which reads:
\begin{multline}
		Z(\beta_L,\beta_R)\, \colonequals \frac{e^{A(\beta_L+\beta_R)}}{\eta(\beta_L)\,\eta(\beta_R)}\Big{[}\left(1-e^{-\beta_L}\right)\left(1-e^{-\beta_R}\right) 
		+\int_{T_{\rm gap}}^\infty dh\int_{T_{\rm gap}}^\infty d\bar{h}\,\rho(h,\bar{h})\,e^{-h\beta_L-\bar{h}\beta_R}\Big{]}\,.
\end{multline}
Here $\eta(\beta)$ is the Dedekind eta function,  which counts the Virasoro descendants, and
\begin{equation}
	A\colonequals \frac{c-1}{24}\,,
\end{equation}
with $c > 1$ the central charge of the theory. In a bona fide CFT the partition function is invariant under the modular transformations generated by:
\begin{equation}
	\begin{split}
		T:&\quad\beta_L,\beta_R\ \rightarrow\ \beta_L+2\pi i,\ \beta_R-2\pi i\,, \\
		S:&\quad\beta_L,\beta_R\ \rightarrow\ \frac{4\pi^2}{\beta_L},\ \frac{4\pi^2}{\beta_R}\,. \\
	\end{split}
\end{equation}
Invariance under $T$ simply reaffirms that the spins $J$ must be integers, but the invariance under $S$ is non-trivial.

In the following we will exclusively work with the \emph{Virasoro primary partition function} $\tilde Z(\beta_L,\beta_R) \colonequals \eta(\beta_L)\eta(\beta_R) Z(\beta_L,\beta_R)$, which reads
\begin{equation}\label{def:ZVir}
	\tilde Z(\beta_L,\beta_R)\, = e^{A(\beta_L+\beta_R)}\Big{[}\left(1-e^{-\beta_L}\right)\left(1-e^{-\beta_R}\right) 
	+\int_{T_{\rm gap}}^\infty dh\int_{T_{\rm gap}}^\infty d\bar{h}\,\rho(h,\bar{h})\,e^{-h\beta_L-\bar{h}\beta_R}\Big{]}\,,
\end{equation}
and for which invariance under $S$ gives the constraint:
\begin{equation}\label{ZVir:modinv}
	\begin{split}
		\tilde{Z}(\beta_L,\beta_R)=\sqrt{\frac{4\pi^2}{\beta_L\beta_R}}\tilde{Z}\left(\frac{4\pi^2}{\beta_L},\frac{4\pi^2}{\beta_R}\right)\,.
	\end{split}
\end{equation}
We will call equation \eqref{def:ZVir} the \emph{direct channel expansion}, and \eqref{def:ZVir} with the replacement $\beta_L,\beta_R\rightarrow\frac{4\pi^2}{\beta_L},\frac{4\pi^2}{\beta_R}$ the \emph{dual channel expansion}.

\subsubsection*{Expectations from the leading-order behavior}
Invariance under $S$ implies that,  when evaluated pointwise in $\beta_L$,  the Virasoro primary partition function diverges exponentially as $\beta_R$ approaches zero.  More precisely,  we can write
\begin{equation}\label{leadingmodular}
\widetilde{Z}(\beta_L,\beta_R) \underset{\beta_R\searrow 0}{\sim}\, \sqrt{\frac{4\pi^2}{\beta_L \beta_R}}\ e^{A \left(\frac{4\pi^2}{\beta_L} + \frac{4\pi^2}{\beta_R}\right)}\left(1-e^{-\frac{4\pi^2}{\beta_L}}\right)\,,
\end{equation}
where $f \sim \, g$ means $\lim f/g = 1$. The $\beta_L$ dependence on the right-hand side of equation \eqref{leadingmodular} is the Laplace transform of a density $\rho_{\rm c}(h)$, defined such that
\begin{equation}
	\label{eq:rho0implicit}
	\sqrt{ \frac{2\pi}{\beta_L}} e^{A \frac{4\pi^2}{\beta_L}} \left( 1 - e^{-\frac{4\pi^2}{\beta_L}} \right)= e^{A \beta_L}\int dh\, \rho_{\rm c}(h) e^{-\beta_L h}\,,
\end{equation}
and explicitly given by:
\begin{equation}\label{def:rho0}
	\begin{split}
		\rho_{\rm c}(h)=\begin{cases}
			\sqrt{\frac{2}{h-A}}\left[\cosh\left(4\pi\sqrt{A(h-A)}\right)-\cosh\left(4\pi\sqrt{(A-1)(h-A)}\right)\right]& h\geqslant A\,, \\
			 &\\
			0& h<A\,. \\
		\end{cases}
	\end{split}
\end{equation}
Equation \eqref{leadingmodular} then naively suggests that
\begin{equation}
	\rho(h,\hb) \overset{?}{\underset{\hb \to \infty}{\leadsto}} \rho_{\rm c}(\hb) \rho_{\rm c}(h) \underset{\hb \to \infty}{\sim} \frac{1}{\sqrt{2\hb}} e^{4 \pi \sqrt{A \hb}}  \rho_{\rm c}(h) \,,
\end{equation}
or, perhaps more intuitively, that
\begin{equation}
	\label{rhoJlargeJintuitive}
	2 \rho_J(J + 2h) \overset{?}{\underset{J \to \infty}{\leadsto}} \frac{1}{\sqrt{2J}} e^{4 \pi \sqrt{A J}} \rho_{\rm c}(h) 
\end{equation}
where we just substituted \eqref{rhohhbdefn} and $\hb = J + h$. This is equation \eqref{rhoJnaiveintro} in the introduction and, as we discussed there, it can only be true in some averaged sense.

\section{The Laplace transform at large spin}
\label{sec:laplacelargespin}
To smoothen out the distributional nature of $\rho_J(2h+J)$ we will integrate it against some test function $\varphi(h)$, like so:
\begin{equation}
	\label{desiredlargeJ}
	2 \int_{T_\text{gap}}^\infty dh\, \varphi(h) \rho_J(2h + J) \,.
\end{equation}
(The factor two ensures that a term $\delta(\Delta - \Delta_k)$ in $\rho_J(\Delta)$ contributes $\varphi(h_k)$ to the integral, with $h_k = (\Delta_k - J)/2$.) The interesting question will now be for which class of test functions $\varphi(h)$ the large $J$ limit is under control. Clearly,  $\varphi(h)$ cannot be a delta function,  since the pointwise large $J$ limit cannot exist.  But can it be a compactly supported function? If so, does it need to be smooth? And would it perhaps be possible to take $\varphi(h)$ to be $J$-dependent, so that its support shrinks with $J$?

We can write equation \eqref{desiredlargeJ} in Fourier space as
\begin{equation}
	\label{eq:largeJwithbeta}
	2 \int \frac{ds}{2\pi} \left[\int dh' \, \varphi(h') e^{(\beta_L + i s) h'} \right] \int_{T_\text{gap}}^\infty dh \, \rho_J(2h + J)  e^{- (\beta_L + i s) h}\,,
\end{equation}
where we introduced an auxiliary parameter $\beta_L$. Clearly the final result does not depend on $\beta_L$, but the integrals over $s$ and $h$ can be swapped only when $\beta_L > 0$. We are thus led to consider the behavior in the right half $\beta_L$ plane of
\begin{equation}
	F(\beta_L,J) \colonequals 2 \int_{T_\text{gap}}^\infty dh\, \rho_J(J + 2h) e^{-\beta_L (h-A)}\,,
\end{equation}
which will be our main object of study in this paper. The next proposition describes two elementary but important properties.

\begin{proposition}
For any fixed $J \in \Z$, $F(\beta_L,J)$ is analytic in the right half plane $\Re(\beta_L) > 0$. In this region it obeys the inequality:
\begin{equation}
	|F(\beta_L,J)| \leq F(\Re(\beta_L),J)\,.
\end{equation}
\end{proposition}

\subsubsection*{Using the partition function}
We would like to write $F(\beta_L,J)$ in terms of the partition function. First we trivially write:
\begin{equation}
	F(\beta_L,J) = 2 \int_{-\pi}^{\pi} \frac{dt}{2 \pi} \, e^{( \alpha + i t) J} \sum_{\tilde J} \int_{T_\text{gap}}^\infty dh\, \rho_{\tilde J}(\tilde J + 2h) e^{-\beta_L (h-A)} e^{-( \alpha + i t) \tilde J}\,,
\end{equation}
where we introduced an auxiliary parameter $\alpha$. Clearly the final result does not depend on $\alpha$, but the sum over $\tilde J$ and integral over $t$ can be swapped only when $0 < \alpha < \beta_L$.

Now let us make two cosmetic changes. First we introduce the complex parameter
\begin{equation}
	z = \alpha + i t\,,
\end{equation}
and write the integral over $t$ as the integral of a contour $C_{\alpha}$ which (for now) goes straight from $\alpha -i \pi$ to $\alpha + i \pi$ in the complex $z$ plane. Second, we write the integral in terms of the density $\rho(h,\hb) = 2 \sum_J \rho_J(h + \hb) \delta(J - \hb -h)$ introduced in the equation \eqref{rhohhbdefn}. This leads to:
\begin{equation}
	F(\beta_L,J) = \int_{C_\alpha} \frac{dz}{2 \pi i} \, e^{z J}  \int_{T_\text{gap}}^\infty dh \int_{T_\text{gap}}^\infty d\hb \, \rho(h,\hb)  e^{- (\beta_L - z) (h-A) -z (\hb-A)}\,,
\end{equation}
The final two integrals yield almost the Virasoro primary partition function $\tilde Z(\beta_L-z,z)$ defined in equation \eqref{def:ZVir}. The mismatch is due to the Virasoro identity block, but since that only contributes at $J = -1, 0,1$ we are allowed to write
\begin{equation}
	F(\beta_L,J) = \int_{C_\alpha} \frac{dz}{2 \pi i} \, e^{z J}  \tilde Z(\beta_L - z, z) \quad \text{for} \quad |J| > 1\,.
\end{equation}
Our interest lies with large $J$, so from now on we will assume $|J| > 1$ even if we do not write it explicitly.

\subsubsection*{Using modular invariance}
Modular invariance now dictates that
\begin{equation}
	F(\beta_L,J) =  \int_{C_\alpha} \frac{dz}{2 \pi i} \, e^{z J}  \sqrt{\frac{4\pi^2}{z(\beta_L - z)}}\tilde Z\left(\frac{4\pi^2}{\beta_L - z}, \frac{4\pi^2}{z}\right)\,,
\end{equation}
which we can split into two parts using the dual channel expansion as:
\begin{align}\label{F:split}
	F(\beta_L,J) &= F_\text{vac}(\beta_L,J) + F_\text{non-vac}(\beta_L,J)\\
	F_\text{vac}(\beta_L,J) &\colonequals  \int_{C_\alpha} \frac{dz}{2 \pi i} \, e^{z J}  \sqrt{\frac{4\pi^2}{z(\beta_L - z)}}e^{\frac{4\pi^2 A}{\beta_L - z} + \frac{4\pi^2 A}{z}} \p{1 - e^{-\frac{4\pi^2}{\beta_L - z}}}\p{1 - e^{-\frac{4\pi^2}{z}}}\,,\nn\\
	F_\text{non-vac}(\beta_L,J) &\colonequals  \int_{C_\alpha} \frac{dz}{2 \pi i} \, e^{z J}  \sqrt{\frac{4\pi^2}{z(\beta_L - z)}}  \int_{T_\text{gap}}^\infty dh \int_{T_\text{gap}}^\infty d\hb \, \rho(h,\hb) e^{- \frac{4\pi^2}{\beta_L - z} (h-A) - \frac{4\pi^2}{z} (\hb -A)}\,.\nn
\end{align}
We already mentioned that $F(\beta_L,J)$ is independent of $\alpha = \Re(z)$, but this is not necessarily the case for $F_\text{vac}(\beta_L,J)$ and $F_\text{non-vac}(\beta_L,J)$ individually. From now on we will therefore fix:
\begin{equation}
	\label{eq:alphachoice}
	\alpha = 2 \pi \sqrt\frac{A}{J} \,,
\end{equation}
and $F_\text{vac}(\beta_L,J)$ and $F_\text{non-vac}(\beta_L,J)$ are always understood to be defined with this value of $\alpha$. This $J$-dependent choice of $\alpha$ requires some discussion.

\subsubsection*{Choosing the optimal $\alpha$}
Our main objective is to obtain the best possible constraints on the behavior at large $J$ of the non-universal term $F_\text{non-vac}(\beta_L,J)$. We do so in the subsection \ref{subsec:Fnonvacestimates}. To illustrate the logic we can look at equation \eqref{eq:Fnonvac1messy} below. This equation is valid for any $0 < \alpha < \Re(\beta_L)$, and provides an upper bound with a factor of the form
\begin{equation}
	\exp\p{\alpha J + \frac{4\pi^2}{\alpha} (A-T_\text{gap})}\,.
\end{equation}
Since $A > T_\text{gap}$, this bound is strongest at large $J$ if $\alpha$ scales like $1/\sqrt{J}$, say
\begin{equation}
	\alpha = 2 \pi \sqrt{\frac{A}{J}} \gamma
\end{equation}
for some finite $\gamma > 0$. (The same scaling is also required to obtain an optimal bound in equation \eqref{eq:Fnonvac2messy}.) The exact choice of $\gamma$ now matters little for our conclusions, but we have set $\gamma = 1$ to yield the optimal bound in the limit where $T_\text{gap}$ becomes very small. Our final choice \eqref{eq:alphachoice} is also the same as in \cite{Pal:2023cgk}, where it was used to discuss the large $J$ spectrum around $h=A$.
%\BvR{Right?}\JQ{Yes.}

We note that the saddle point approximation for the vacuum term in the proof of lemma \ref{lem:Fvaclimit}(b) yields the same answer for all $0 < \gamma < 2$.

\subsection{Estimates for \texorpdfstring{$F_{\mathrm{vac}}(\beta_L,J)$}{Fvac(betaL,J)}}
\label{subsec:Fvacestimates}
The contribution of the dual channel vacuum term $F_\text{vac}(\beta_L,J)$ is a relatively straightforward integral. We will need two properties.

\begin{lemma}
\label{lem:Fvaclimit}
For any $\beta_L$ in the right half plane we have
\begin{align}
	&(a) & \left| F_{\rm{vac}}\p{\beta_L, J} \right| \leq &\frac{1}{\sqrt{2 J}} e^{4 \pi \sqrt{A J}} \frac{C_{\Re(\beta_L)}}{(\,1 + |\Im(\beta_L)|\,)^{3/2}}\\
	&(b) & F_\text{vac}\left(\beta_L,J \right) \underset{J \to \infty}{\sim}\, &\frac{1}{\sqrt{2 J}} e^{4 \pi \sqrt{A J}}\sqrt{\frac{2\pi}{\beta_L}}e^{\frac{4\pi^2 A}{\beta_L}} \p{1 - e^{-\frac{4\pi^2}{\beta_L}}}\,,
\end{align}
where the inequality holds for sufficiently large $J$.
\end{lemma}
This lemma in particular says that the leading large $J$ growth of $F_\text{vac}(\beta_L,J)$ is of the form
\begin{equation}
	\label{largeJtermFvac}
	e^{4 \pi \sqrt{AJ}}/\sqrt{2J}
\end{equation}
It is essential to keep this behavior in mind for the remainder of the paper. In particular, meaningful conclusions about the large $J$ limit can only be obtained if the non-universal term $F_\text{non-vac}(\beta_L,J)$ can be shown to be $o(e^{4 \pi \sqrt{AJ}}/\sqrt{2J})$ at large $J$. We will provide a much better estimate in the next subsection, but let us first prove the lemma.

\begin{proof}[Proof of lemma \ref{lem:Fvaclimit}(a)]
	Let us set $\beta_L = \beta + i s$ with $\beta > 0$ and $s \in \R$. Starting from the definition in equation \eqref{F:split} we easily obtain the estimate:
	\begin{equation}\label{Fvac:generalbound}
		\begin{split}
			\left|F_\text{vac}\p{\beta + i s,J}\right| &\leq 2 e^{\alpha J} \alpha^{-1/2} e^{\frac{4\pi^2 A}{\beta - \alpha}} \,\left(\frac{8\pi^2\left[1+\pi+2(\beta-\alpha)\right]}{(\beta-\alpha)(1+\abs{s})}\right)^{3/2}\int_{-\pi}^\pi \frac{dt}{2\pi}   e^{\frac{4 \pi^2 A \alpha}{\alpha^2 + t^2}} .
		\end{split}
	\end{equation}
	Here we used that $\abs{e^{x}}\leqslant e^{{\rm Re}(x)}$, $\alpha\leqslant|z|$, and $\abs{1-e^{-4\pi^2/z}}\leqslant 2$ along the integration contour. The only subtle term is:
	\begin{equation}
		\begin{split}
			\abs{\sqrt{\frac{4\pi^2}{\beta_L - z}} \p{1 - e^{-\frac{4\pi^2}{\beta_L - z}}}}&\leqslant\abs{\frac{4\pi^2}{\beta_L-z}}^{3/2}\leqslant\left(\frac{8\pi^2\left[1+\pi+2(\beta-\alpha)\right]}{(\beta-\alpha)(1+\abs{s})}\right)^{3/2}
		\end{split}
	\end{equation}
	which again holds along the $z$ integration contour\footnote{The second inequality can be proven by considering two cases:  1) $|s|\leq 2\pi+1+2(\beta-\alpha)$ and showing that $|\beta_L-z|\geqslant (\beta-\alpha)\geqslant \frac{(\beta-\alpha)(1+\abs{s})}{2\left[1+\pi+2(\beta-\alpha)\right]}$   and 2) for the complementary regime $|\beta_L-z|\geqslant (|s|-\pi) \geqslant \frac{1+|s|}{2}\geqslant  \frac{(\beta-\alpha)(1+\abs{s})}{2\left[1+\pi+2(\beta-\alpha)\right]}$. }.  We recall that $\alpha = 2 \pi \sqrt{A/J}$ so for large $J$ we have $\alpha \to 0$. The remaining integral then limits to its saddle point value
	\begin{equation}
		\int_{-\pi}^\pi \frac{dt}{2\pi}   e^{\frac{4 \pi^2 A \a}{\a^2 + t^2}} \underset{J \to \infty}{\sim} \frac{\alpha^{3/2}}{4\pi^{3/2}\sqrt{A}} e^{\frac{4\pi^2 A}{\a}}\,,
	\end{equation}
	so if we take $J$ sufficiently large it will certainly be less than twice this value. We arrive at:
	\begin{equation}
		\begin{split}
			\left|F_\text{vac}\p{\beta + i s,J}\right| &\leq  \frac{\alpha}{\pi^{3/2}\sqrt{A}} \,\left(\frac{8\pi^2\left[1+\pi+2(\beta-\alpha)\right]}{(\beta-\alpha)(1+\abs{s})}\right)^{3/2}e^{\alpha J+\frac{4\pi^2A}{\alpha}+\frac{4\pi^2 A}{\beta - \alpha}}\\
			&\leqslant C_\beta\,J^{-1/2}\,e^{4\pi\sqrt{AJ}}(1+\abs{s})^{-3/2}
		\end{split}
		\end{equation}
	with
	\begin{equation}
		C_{\beta}=128\,\pi^{5/2}\,\left(1+\frac{1+\pi}{\beta}\right)^{3/2}e^{\frac{8\pi^2A}{\beta}}\,,
	\end{equation}
	where we also required that $J \geq 16\pi^2 A/ \beta^2$.
\end{proof}

\begin{figure}[t]
%	 With[{A = 1/320, \[Gamma] = 1.46}, 
% Show[RegionPlot[
%   x + (4 \[Pi]^2 A x)/(x^2 + y^2) < 4 \[Pi] Sqrt[A], {x, 0, 
%    1}, {y, -4, 4}, PerformanceGoal -> 10, FrameTicks -> None, 
%   PlotStyle -> LightGray, BoundaryStyle -> Directive[Gray, Dashed]], 
%  ContourPlot[x^2 + y^2 == 4 \[Pi]^2 A, {x, 0, 1}, {y, -4, 4}, 
%   ContourStyle -> Blue], 
%  ContourPlot[ 
%   x^2 + y^2 == 4 \[Pi]^2 A, {x, \[Pi] Sqrt[A], 1}, {y, -4, 4}, 
%   ContourStyle -> Red], 
%  Graphics[{Red, Thick, 
%    Line[{{\[Pi] Sqrt[A], 
%       Sqrt[3] \[Pi] Sqrt[A]}, {\[Gamma]  \[Pi] Sqrt[A], 4}}]}], 
%  Graphics[{Red, Thick, 
%    Line[{{\[Pi] Sqrt[A], -Sqrt[3] \[Pi] Sqrt[ 
%         A]}, {\[Gamma]  \[Pi] Sqrt[A], -4}}]}], 
%  Graphics[{Blue, Thick, Disk[{2 \[Pi] Sqrt[A], 0}, {0.01, 0.08}]} ], 
%  Graphics[{Thick, Line[{{0.85, 3.75}, {0.85, 3.4}, {0.9, 3.4}}], 
%    Inset[MaTeX["w", FontSize -> 21], {0.88, 3.6}]}, 
%   PlotRange -> {{0, 1}, {-4, 4}}, Axes -> False, Frame -> False, 
%   ImageSize -> 100] ]]
	\centering
	\includegraphics[width=0.5\textwidth]{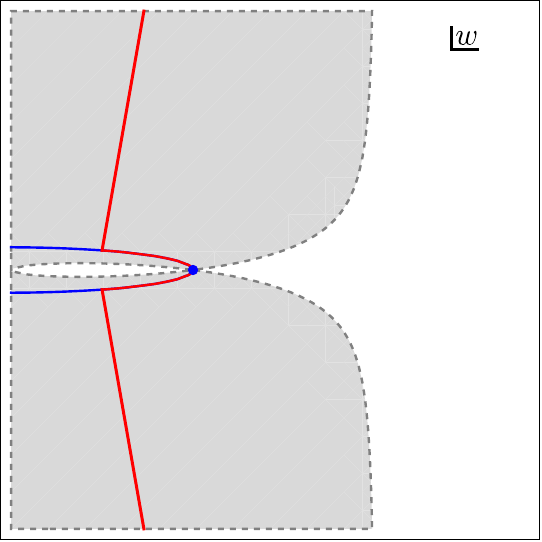} 
	\caption{In red we show the integration contour in the complex $w$ plane. It follows the steepest descent contour, in blue,  in the vicinity of the saddle point, marked with blue dot. Along the rest of the contour the integrand is exponentially suppressed because it lies entirely in the gray shaded region,  which is given by equation \eqref{exponentiallysmall}.}
	\label{fig:contour}
\end{figure}

\begin{proof}[Proof of lemma \ref{lem:Fvaclimit}(b)]
We will show the lemma by using a saddle point approximation to the full integral. We again start from the definition of $F_{\rm vac}(\beta_L, J)$ given in equation \eqref{F:split}. After substituting \(\alpha=2\pi\sqrt{A/J}\) and changing variables to \(w = \sqrt{J} z\) we obtain
\begin{equation}
	\begin{split}
		F_{\rm vac}\left(\beta_L,J\right) &= \frac{1}{J^{1/4}}  \int_{2\pi\sqrt{A}-i\pi\sqrt{J}}^{2\pi\sqrt{A}+i\pi\sqrt{J}} \frac{dw}{2\pi i} \,e^{\sqrt{J}f(w)}\,g(\beta_L,w,J),  \\
		f(w) &= w+\frac{4\pi^2A}{w}, \\
		g(\beta_L,w,J) &= \sqrt{\frac{4 \pi^2}{w(\beta_L-w/\sqrt{J})}} e^{\frac{4\pi^2A}{\beta_L-w/\sqrt{J}}} \left(1 - e^{-\frac{4\pi^2}{\beta_L - w/\sqrt{J}}}\right)\left(1 - e^{-\frac{4\pi^2\sqrt{J}}{w}}\right).
	\end{split}
\end{equation}
The integrand develops a saddle point when $f'(w_*) = 0$, which corresponds to:
\begin{equation}
	w_* = 2 \pi \sqrt{A}\,,
\end{equation}
and the steepest descent contour is given by a circle,
\begin{equation}
	|w|^2 = 4 \pi^2 A\,.
\end{equation}
We can freely deform the $w$ integration contour within the strip
\begin{equation}
	\label{Fvacstrip}
	0 < \Re(w) \ll \Re(\beta_L) \sqrt{J}\,,
\end{equation}
where both \(f(w)\) and \(g(\beta_L,J,w)\) are holomorphic in \(w\). We will take the contour to consist of two parts: a part $C_1$ that follows the steepest descent contour for a finite amount, say until $\Re(w)$ has decreased from $2 \pi \sqrt{A}$ to $\pi \sqrt{A}$, and another part $C_2$ that connects this circular arc straight to the endpoints at $w = 2 \pi \sqrt{A} \pm i \pi \sqrt{J}$. See figure \ref{fig:contour}.

The $C_1$ segment contains $w_*$. Also, $g(\beta_L,w,J)$ converges uniformly to its limit $g(\beta_L,w,\infty)$, so we may use the saddle point approximation. This means that the large $J$ contribution from $C_1$ is given by
\begin{equation}
	\frac{e^{\sqrt{J} f(w_*)}}{\sqrt{2\pi J |f''(w_*)|}} g(\beta_L,w_*,\infty) = \frac{1}{\sqrt{2 J}} e^{4 \pi \sqrt{A J}}\sqrt{\frac{2\pi}{\beta_L}}e^{\frac{4\pi^2 A}{\beta_L}} \p{1 - e^{-\frac{4\pi^2}{\beta_L}}}\,,
\end{equation}
which is the limit as claimed in the statement of the lemma.

It remains to show that the contribution from $C_2$ is subleading. Since we are in the strip given by equation \eqref{Fvacstrip}, we can first of all estimate
\begin{equation}
	|g(\beta_L, w, J)| \leq 4 \sqrt{\frac{4 \pi^2}{|w| (\Re(\beta_L)-\Re(w)/\sqrt{J})}} e^{\frac{4\pi^2A}{\Re(\beta_L)-\Re(w)/\sqrt{J}}}\,,
\end{equation}
where we also used that $|1 - e^{-1/z}| < 2$ if $\Re(z) > 0$. Using also that $|w| \geq 2 \pi \sqrt{A}$, we conclude that $|g(\beta_L, w, J)|$ is uniformly bounded along $C_2$. We can then achieve an exponential suppression of the integrand if $C_2$ lies entirely in the region where
\begin{equation}
	\label{exponentiallysmall}
	\left|e^{\sqrt{J} f(w)}\right| < e^{\sqrt{J} f(w_*)}\,,
\end{equation}
which is the gray shaded region in figure \ref{fig:contour}. Graphically this is seen to be possible, but only if the endpoints at $w = 2 \pi \sqrt{A} \pm i \pi \sqrt{J}$ also lie inside the gray region. We therefore need
\begin{equation}
	\left.\Re\p{ w + \frac{4 \pi^2 A}{w}}\right|_{w = 2 \pi \sqrt{A} \pm i \pi J} < 4 \pi \sqrt{A}\,,
\end{equation}
which is true if $J$ is sufficiently large. The length of the $C_2$ segment grows like $\sqrt{J}$ which is not enough to offset the exponential suppression, so the lemma follows.
\end{proof}

\subsection{First bound on \texorpdfstring{$F_\text{non-vac}(\beta_L,J)$}{Fnonvac(betaL,J)}}
\label{subsec:Fnonvacestimates}
We recall the definition of the non-vacuum term in equation \eqref{F:split}:
\begin{equation}
\begin{split}
		\label{Fnonvacrepeated}
		F_\text{non-vac}(\beta_L,J) &\colonequals \int_{C_\alpha} \frac{dz}{2 \pi i} \, e^{z J}  \sqrt{\frac{4\pi^2}{z(\beta_L - z)}}  \int_{T_\text{gap}}^\infty dh \int_{T_\text{gap}}^\infty d\hb \, \rho(h,\hb) e^{- \frac{4\pi^2}{\beta_L - z} (h-A) - \frac{4\pi^2}{z} (\hb -A)}\,.
\end{split}
\end{equation}
As before, $z = \a + i t$ with $\alpha = 2 \pi \sqrt{A/J}$, and the contour $C_\alpha$ corresponds to an integral over $t$ from $-\pi$ to $\pi$.  In this subsection we will prove the following bound.

\begin{proposition}\label{prop:Fnonvaccomplex}
	For any $\beta>0$ and sufficiently large $J$,
	\begin{equation}
		\begin{split}
			\abs{F_{\rm{non-vac}}\left(\beta+ is,J\right)}&\leqslant C_\beta\,
			\sqrt{J}\,e^{4\pi \sqrt{AJ} \left(1 - \tau\right) +\frac{2A}{\beta} s^2},
		\end{split}
	\end{equation}
	where
	\begin{equation}
		\tau =\min\left\{\frac{T_{\rm gap}}{2A}, \frac{3}{8}\right\}\,. \\
	\end{equation}
	and $C_\beta$ is independent of $s$ and $J$.
\end{proposition}

The fact that $\tau > 0$ ensures that the large $J$ growth of $F_\text{non-vac}(\beta_L,J)$ is significantly slower than that of $F_\text{vac}(\beta_L,J)$, which was $e^{4\pi \sqrt{AJ}}/\sqrt{2J}$ as we discussed in the previous subsection. The price we had to pay was the introduction of the factor $e^{\frac{2A}{\beta} s^2}$. It has rapid growth at large $|s|$,  which will somewhat complicate our analysis in the next section. Note also that the estimate changes when $T_\text{gap}$ crosses $3A/4$, but this value might merely be an artifact of our approximations without physical significance.\footnote{We point out that the quantity $3A/4$ appears also in \cite{Benjamin:2019stq, Kusuki:2018wpa,Collier:2018exn}. }

Bounding $F_\text{non-vac}(\beta_L,J)$ starts with the elementary observation that, given a real $\mu$,
\begin{equation}
	\label{eq:supexp}
	\sup_{t \in (-\pi,\pi)}\left|\exp\p{- \mu \frac{4\pi^2}{z}} \right| = \sup_{t \in (-\pi,\pi)} \exp\p{-\mu \frac{4\pi^2 \alpha}{\alpha^2 + t^2}} = 
	\begin{cases}
	\exp( - \mu \frac{4\pi^2}{\a}) &\quad \text{if } \mu \leq 0\,,\\
	\exp( - \mu \frac{4\pi^2\a}{\a^2 + \pi^2}) &\quad \text{if } \mu \geq 0\,.
	\end{cases}
\end{equation}
The two different possibilities motivate splitting the $\hb$ integral in equation \eqref{Fnonvacrepeated} at $A$, so we write:
\begin{align}
	\label{eq:Fnonvacsplit}
	F_\text{non-vac}(\beta_L,J,\alpha) &= F_\text{non-vac}^{(1)}(\beta_L,J,\alpha) + F_\text{non-vac}^{(2)}(\beta_L,J,\alpha)\,,\\
	F_\text{non-vac}^{(1)}(\beta_L,J,\alpha) &\colonequals  \int_{C_\alpha} \frac{dz}{2 \pi i} \, e^{z J}  \sqrt{\frac{4\pi^2}{z(\beta_L - z)}}  \int_{T_\text{gap}}^\infty dh \int_{T_\text{gap}}^A d\hb \, \rho(h,\hb) e^{- \frac{4\pi^2}{\beta_L - z} (h-A) - \frac{4\pi^2}{z} (\hb -A)}\,,\nn \\
	F_\text{non-vac}^{(2)}(\beta_L,J,\alpha) &\colonequals  \int_{C_\alpha} \frac{dz}{2 \pi i} \, e^{z J}  \sqrt{\frac{4\pi^2}{z(\beta_L - z)}}  \int_{T_\text{gap}}^\infty dh \int_{A}^\infty d\hb \, \rho(h,\hb) e^{- \frac{4\pi^2}{\beta_L - z} (h-A) - \frac{4\pi^2}{z} (\hb -A)}\,. \nn
\end{align}
We bound each term separately in the next two lemmas. Proposition \ref{prop:Fnonvaccomplex} is then easily found by combining lemma \ref{lemma:Fnonvac1} and lemma \ref{lemma:Fnonvac2}. (For later convenience we replaced the factors $e^{\frac{9A}{8\beta} (|s| + \pi)^2}\sqrt{\beta^2 +2 (\abs{s}+\pi)^2}$ in the lemmas with $C e^{\frac{2A}{\beta} s^2}$ in the proposition.)

\begin{lemma}\label{lemma:Fnonvac1}
	% For $J \geqslant \max\left\{ \frac{400\pi^2 A^2}{\beta^2},\, A \right\} $, we have
	For sufficiently large $J$,
	\begin{equation}
		\begin{split}
			\left|F_\text{non-vac}^{(1)}\left(\beta+ is, J \right)\right| 
			&\leq C^{(1)}_\beta \, J^{1/4}\,\sqrt{\beta^2 + 2(\abs{s}+\pi)^2}\,
			e^{4\pi \sqrt{AJ} \left(1 - \frac{T_{\rm gap}}{2A} \right) + \frac{9A}{8\beta} (|s|+\pi)^2},
		\end{split}
	\end{equation}
	where $C^{(1)}_\beta$ is independent of $s$ and $J$.
\end{lemma}

\begin{proof}
	In the definition of equation \eqref{eq:Fnonvacsplit} we substitute $z = \alpha + i t$ and bound the integral by the supremum of the absolute value of the integrand, for $t \in [-\pi,\pi]$. To a first approximation this yields:
	\begin{multline}
			\label{Fnonvac1firstbound}
			|F_\text{non-vac}^{(1)}(\beta+ is,J)|\leq \\ e^{\a J}  \sqrt{\frac{4\pi^2}{\alpha(\beta - \alpha)}}\int_{T_\text{gap}}^\infty dh \int_{T_\text{gap}}^A d\hb \, \rho(h,\hb) e^{- \frac{4\pi^2}{\alpha} (\hb -A)} \sup_{t \in [-\pi,\pi]} \left|e^{- \frac{4\pi^2}{\beta + i s - \alpha - i t} (h-A)}\right|
	\end{multline}
	where we used for example that $|z| > \Re(z)$ and also our elementary observation \eqref{eq:supexp}. For the final term we use
	\begin{equation}
		\begin{split}
		&\sup_{t \in [-\pi,\pi]} \left|\exp\p{- \frac{4\pi^2}{\beta + i s - \alpha - i t} (h-A)} \right| =\sup_{t \in [-\pi,\pi]} \exp\p{- \frac{4\pi^2 (\beta - \alpha)}{(\beta - \alpha)^2 + (s+ t)^2} (h-A)} \\
		&\qquad \leq \exp\p{- \frac{4\pi^2 (\beta - \alpha)}{(\beta - \alpha)^2 + (|s|+ \pi)^2} h + \frac{4\pi^2}{\beta - \alpha}A}\\
		&\qquad = \exp\p{- \frac{4\pi^2 (\beta - \alpha)}{(\beta - \alpha)^2 + (|s|+ \pi)^2} (h- A)}\exp\p{ \frac{4\pi^2}{\beta - \alpha}A - \frac{4\pi^2 (\beta - \alpha)}{(\beta - \alpha)^2 + (|s|+ \pi)^2} A}
		\end{split}
	\end{equation}
	where we used \eqref{eq:supexp} again to go to the second line. As will soon become clear, we also need to rewrite the other exponential term in equation \eqref{Fnonvac1firstbound}. We use $T_\text{gap} \leq \hb \leq A$ and $\alpha < 2\pi$ (for sufficiently large $J$) to write:
	\begin{equation}
	\begin{split}
		\exp\p{- \frac{4\pi^2}{\alpha}(\hb - A)} &= \exp\p{- \frac{4\pi^2}{\alpha}(T_\text{gap} - A) - \frac{4\pi^2}{\alpha}(\hb - T_\text{gap})} \\
		&\leq \exp\p{- \frac{4\pi^2}{\alpha}(T_\text{gap} - A) - 2\pi (\hb - T_\text{gap})}\\
		&= \exp\p{- \frac{4\pi^2}{\alpha}(T_\text{gap} - A) - 2\pi (A - T_\text{gap})} \exp\p{- 2\pi (\hb - A) }\,.
	\end{split}
	\end{equation}
	After substitution of both these bounds we can bound the remaining $h$ and $\hb$ integrals by the full Virasoro primary partition function. We find the somewhat messy estimate:
	\begin{multline}
		\label{eq:Fnonvac1messy}
		|F_\text{non-vac}^{(1)}(\beta+ is,J)|\leq\\
		\sqrt{\frac{4\pi^2}{\alpha(\beta - \alpha)}} e^{\alpha J - \left(\frac{4\pi^2}{\alpha} - 2\pi\right) (T_\text{gap} - A)  + \frac{4\pi^2}{\beta - \alpha}A - \frac{4\pi^2 (\beta - \alpha)}{(\beta - \alpha)^2 + (|s|+ \pi)^2} A} \tilde Z \left(\frac{4\pi^2 (\beta - \alpha)}{(\beta - \alpha)^2 + (|s|+ \pi)^2} , 2 \pi \right)\,.
	\end{multline}
	The trick is now to use modular invariance again, leading to the equivalent bound:
	\begin{multline}
		|F_\text{non-vac}^{(1)}(\beta+ is,J)|\leq\\
		\sqrt{\frac{2\pi((\beta - \alpha)^2 + (|s|+ \pi)^2)}{\alpha(\beta - \alpha)^2}} e^{\alpha J - \left(\frac{4\pi^2}{\alpha} - 2\pi\right) (T_\text{gap} - A)  + \frac{4\pi^2}{\beta - \alpha}A - \frac{4\pi^2 (\beta - \alpha)}{(\beta - \alpha)^2 + (|s|+ \pi)^2} A} \tilde Z \left(\frac{(\beta - \alpha)^2 + (|s|+ \pi)^2}{\beta - \alpha} , 2 \pi \right)\,.
	\end{multline}
	Now recall that $\alpha = 2 \pi \sqrt{A/J} \ll \beta$, so for sufficiently large $J$ we can clean up the prefactor to obtain:
	\begin{multline}
		\label{eq:Fnonvac1refined}
		|F_\text{non-vac}^{(1)}(\beta+ is,J)|\leq\\
		C_\beta J^{1/4} \sqrt{\beta^2 + 2 (|s|+ \pi)^2} e^{4 \pi \sqrt{A J} \left( 1 -  \frac{T_\text{gap}}{2 A} \right)} \tilde Z \left(\frac{(\beta - \alpha)^2 + (|s|+ \pi)^2}{\beta - \alpha} , 2 \pi \right)\,.
	\end{multline}

	This leaves us with an estimate for the partition function itself. As follows directly from its definition, it obeys the inequality
	\begin{equation}\label{partition:boundvac}
		\widetilde{Z}(\beta_L, \beta_R) \leqslant 
		\left(1 + e^{-4\pi A} \widetilde{Z}(2\pi, 2\pi) \right) 
		e^{A(\beta_L + \beta_R)} \qquad  (\beta_L, \beta_R \geqslant 2\pi)\,,
	\end{equation}
	which implies that
	\begin{equation}
		\widetilde{Z}\left( \frac{(\beta - \alpha)^2 + (|s| + \pi)^2}{\beta - \alpha},\, 2\pi \right) 
		\leqslant \left(1 + e^{-4\pi A} \widetilde{Z}(2\pi, 2\pi) \right)
		e^{A(\beta + 2\pi) + \frac{9A}{8\beta} (\abs{s}+\pi)^2 }\,,
	\end{equation}
	provided $\alpha \leqslant \beta / 9$, which is certainly the case for sufficiently large $J$. Combining this with equation \eqref{eq:Fnonvac1refined} yields the lemma, with
	\begin{equation}
		C_\beta^{(1)} = C_\beta \left(1 + e^{-4\pi A} \widetilde{Z}(2\pi, 2\pi) \right)
		e^{A(\beta + 2\pi)}\,.
	\end{equation}
\end{proof}

\begin{lemma}\label{lemma:Fnonvac2}
For sufficiently large $J$,
	\begin{equation}
		\begin{split}
			\abs{F_\text{non-vac}^{(2)}\left(\beta+ is,J,2\pi\sqrt{\frac{A}{J}}\right)}&\leqslant C_\beta^{(2)} J^{1/2}\,\sqrt{\b^2+ 2(\abs{s}+\pi)^2}e^{\frac{5\pi}{2}\sqrt{AJ}+\frac{9A}{8\beta}(\abs{s}+\pi)^2 },
		\end{split}
	\end{equation}
	where $C_\beta^{(2)}$ is independent of $s$ and $J$.
\end{lemma}
\begin{proof}
	We recall the definition:
	\begin{equation}
		F_\text{non-vac}^{(2)}(\beta_L,J,\alpha) \colonequals  \int_{C_\alpha} \frac{dz}{2 \pi i} \, e^{z J}  \sqrt{\frac{4\pi^2}{z(\beta_L - z)}}  \int_{T_\text{gap}}^\infty dh \int_{A}^\infty d\hb \, \rho(h,\hb) e^{- \frac{4\pi^2}{\beta_L - z} (h-A) - \frac{4\pi^2}{z} (\hb -A)} \\
	\end{equation}
	With the same steps as in the proof of lemma \ref{lemma:Fnonvac1}, we obtain:
	\begin{equation}
		\label{eq:Fnonvac2messy}
		\begin{split}
			&|F_\text{non-vac}^{(2)}(\beta+ is,J,\alpha)| \\
			&\leq e^{\a J} e^{\frac{4\pi^2}{\b - \a}A- \frac{4\pi^2 (\b - \a)}{(\b - \a)^2 + (|s| +\pi)^2} A } \sqrt{\frac{4\pi^2}{\alpha(\b - \a)}} \tilde Z\left(\frac{4\pi^2 (\b - \a)}{(\b - \a)^2 + (|s| +\pi)^2}, \frac{4\pi^2 \a}{\a^2 + \pi^2}  \right)\\
			&=e^{\a J}  e^{\frac{4\pi^2}{\b - \a}A- \frac{4\pi^2 (\b - \a)}{(\b - \a)^2 + (|s| +\pi)^2} A }\sqrt{\frac{(\b - \a)^2 + (|s| +\pi)^2}{(\b - \a)^2} \,\, \frac{\a^2 + \pi^2}{\a^2}} \tilde Z\left(\frac{(\b - \a)^2 + (|s| +\pi)^2}{(\b - \a)}, \frac{\a^2 + \pi^2}{\a}  \right).
		\end{split}
	\end{equation}
	Compared to lemma \ref{lemma:Fnonvac1}, the only difference here is that the range of $\bar{h}$ is $\bar{h}\geqslant A$, which makes the second argument of the partition function slightly different. The lemma follows after applying \eqref{partition:boundvac} and cleaning up the prefactor, using that $\alpha=2\pi\sqrt{A/J}\leqslant\beta/9$ for large enough $J$.  
\end{proof}

\subsection{Second bound on \texorpdfstring{$F_\text{non-vac}(\beta_L,J)$}{Fnonvac(betaL,J)}}
Below we will need one more bound on the non-vacuum term.

\begin{proposition}
	For sufficiently large $J$,
	\begin{equation}
	\abs{F_\text{non-vac}(\beta + is,J)} \leq \frac{1}{\sqrt{2J}} e^{4 \pi \sqrt{AJ}} C''_\beta
	\end{equation}
	where $C''_\beta$ is independent of $s$ and $J$.
	\end{proposition}
	\begin{proof}
	We use $F_\text{non-vac}(\beta_L,J) = F(\beta_L,J) - F_\text{vac}(\beta_L,J)$ to write, again for large enough $J$,
	\begin{equation}
		\begin{split}
			|F_\text{non-vac}(\beta + is,J)| &\leq |F(\beta+is,J)| + |F_\text{vac}(\beta + is,J)| \\ 
			&\leq F(\beta,J) + \frac{1}{\sqrt{2J}} e^{4 \pi \sqrt{AJ}} \frac{C_\beta}{(1 + |s|)^{3/2}}\\
			&= F_\text{vac}(\beta,J) + F_\text{non-vac}(\beta,J) + \frac{1}{\sqrt{2J}} e^{4 \pi \sqrt{AJ}} \frac{C_\beta}{(1 + |s|)^{3/2}}\\
			&\leq 2 C_\beta \frac{1}{\sqrt{2J}} e^{4 \pi \sqrt{AJ}} + C_\beta' \sqrt{J} e^{4 \pi \sqrt{AJ}(1 -\tau)}\,.
		\end{split}
	\end{equation}
	where we used lemma \ref{lem:Fvaclimit}(a) to go to the second line and proposition \ref{prop:Fnonvaccomplex} to go to the last line. The proposition follows since $\tau > 0$.
\end{proof}

\section{A theorem for the modular bootstrap}
\label{sec:theorem}
Let us define the spin-$J$ density of the vacuum term as the inverse Laplace transform of $F_\text{vac}(\beta_L,J)$:
\begin{equation}
	\label{rhoJvacdefn}
	2 \rho_{J,\text{vac}}(J + 2h) \colonequals \int \frac{ds}{2\pi} F_\text{vac}(\beta + i s, J) e^{(\beta + is)(h-A)}\,.
\end{equation}
Here we recall the definition of $F_\text{vac}(\beta_L,J)$ in \eqref{F:split}. We stress that this is the \emph{universal} part of the large-spin spectral density. The large $J$ expansion of $\rho_{J,\text{vac}}(J+ 2h)$ is entirely calculable and theory-independent. Our objective in this section is not to calculate it in detail, but rather to put an upper bound on the \emph{theory-dependent} terms that come from $F_\text{non-vac}(\beta_L,J)$.

\begin{remark}
The direct channel expansion of the vacuum in the dual channel is $\rho_{\rm c}(h) \rho_{\rm c}(\hb)$, as we showed in equation \eqref{eq:rho0implicit}. This does not have a decomposition into integer spins, so the ``spin-$J$ projection" of this term does not exist. As a simple computation shows, equation \eqref{rhoJvacdefn} amounts to the specific choice:
\begin{equation}
	\rho_{J,\text{vac}}(J + 2h) \colonequals \rho_{\rm c}(h) \int d\hb \, \rho_{\rm c}(\hb) e^{\alpha (h-\hb+J)} \frac{ \sin (\pi  (J + h-\hb))}{\pi  (J + h-\hb)}\,,
\end{equation}
with, as always, $\alpha = 2 \pi \sqrt{A/J}$.
\end{remark}

In the next proposition we show that the leading large $J$ term of $\rho_{J,\text{vac}}(J + 2h)$ has the expected behavior.

\begin{proposition}
	\label{prop:rhovaclimituniform}
	The limit
	\begin{equation}
		\lim_{J \to \infty} 2 \sqrt{2J} e^{- 4 \pi \sqrt{AJ}} \rho_{J,\rm{vac}}(J + 2h) = \rho_{\rm c}(h)\,,
	\end{equation}
	is uniform for $h$ in any finite interval.
\end{proposition}

\begin{remark}
The formatting used in the rest of the paper would have suggested the notation
\begin{equation}
	2 \rho_{J,\rm{vac}}(J + 2h) \underset{J \to \infty}{\sim} \frac{1}{\sqrt{2J}} e^{ 4\pi \sqrt{AJ}} \rho_{\rm c}(h)\,,
\end{equation}
but here this is not quite possible: our definition that $a \sim b$ equals $\lim a/b = 1$ does not work for $h \leq A$ since in that region $\rho_{\rm c}(h) = 0$.
\end{remark}

\begin{proof}
	Fix some $\beta > 0$. The definitions \eqref{rhoJvacdefn} and \eqref{eq:rho0implicit} inform us that:
	\begin{equation}
		\begin{split}
			2\rho_{J,\text{vac}}(2h + J)e^{-\beta(h-A)} &= \int\frac{ds}{2\pi}F_\text{vac}\left(\beta+is,J,2\pi\sqrt{\frac{A}{J}}\right) e^{is(h-A)}\\
			\rho_{\rm c}(h) e^{-\beta(h-A)} &= \int \frac{ds}{2\pi}  \sqrt{\frac{2\pi}{\beta + i s}}e^{\frac{4\pi^2 A}{\beta + i s}} \p{1 - e^{-\frac{4\pi^2}{\beta + i s}}} e^{is(h-A)}\,.
		\end{split}
	\end{equation}
	Therefore,
	\begin{equation}
		\begin{split}
			&\left| 2\sqrt{2J}e^{-4\pi\sqrt{AJ}}\rho_{J,\text{vac}}(2h + J) - \rho_{\rm c}(h)\right| e^{-\beta(h-A)} \\
			&\leq \int\frac{ds}{2\pi} \left| \sqrt{2J}e^{-4\pi\sqrt{AJ}} F_\text{vac}\left(\beta+is,J,2\pi\sqrt{\frac{A}{J}}\right) - \sqrt{\frac{2\pi}{\beta + i s}}e^{\frac{4\pi^2 A}{\beta + i s}} \p{1 - e^{-\frac{4\pi^2}{\beta + i s}}}\right|\,,
		\end{split}
	\end{equation}
	where we made the right-hand side independent of $h$ by using $|e^{i s(h-A)}| = 1$. At large $J$ lemma \ref{lem:Fvaclimit}(b) says that the integrand vanishes pointwise in $s$. But this suffices to prove that the integral also vanishes, because lemma \ref{lem:Fvaclimit}(a) and the dominated convergence theorem tell us that we can swap the large $J$ limit and the integral over $s$. The only residual non-uniformity in $h$ is due to the additional factor $e^{-\beta(h-A)}$ on the left-hand side, but this is bounded if $h$ is restricted to any finite interval.
\end{proof}

Now we discuss the test functions with which we can average the spectral density $\rho_J$. We denote as $\mathcal{D}([-R,R])$ the space of functions $\varphi(h)$ which are smooth and compactly supported in the interval $[-R,R]$.\footnote{An example of such function is $\varphi(h)=\exp\left(\frac{1}{h^2-R^2}\right)$ if $|h|<R$ and $0$ otherwise.} The compact support of $\varphi$ implies that its Fourier transform
\begin{equation}
	\hat \varphi(s) = \int dh \, \varphi(h) e^{i s h}
\end{equation}
is an entire function of $s\in\mathbb{C}$.

For $h_* > 0$ we define the rescaled and translated function:
\begin{equation}\label{def:phirescale}
	\begin{split}
		\varphi_{h_*,\lambda}(h):=\frac{1}{\lambda}\varphi\p{\frac{h-h_*}{\lambda}},\quad\lambda>0\,,
	\end{split}
\end{equation}
which is smooth and supported in the interval $[h_*-\lambda R,h_*+\lambda R]$. Note that
\begin{equation}\label{FT:rescaledphi}
	\begin{split}
		\hat{\varphi}_{h_*,\lambda}(s)=e^{ish_*}\hat{\varphi}(\lambda s)\,,
	\end{split}
\end{equation}
again for $s \in \C$. The most precise limit is when $\lambda \to 0$, we essentially obtain $\hat{\varphi}(0)$ times a delta function $\delta(h-h_*)$. The following theorem tells us that we can send $\lambda \to 0$ almost as fast as $J^{-1/4}$ as $J \to \infty$.

\begin{theorem}\label{theorem:error}
	Suppose $\varphi \in \mathcal{D}([-R,R])$ and pick $\lambda_J$ such that
	\begin{equation}
		\lambda_J \geq \frac{\delta}{J^{\frac{1}{4} - \epsilon} }\,,
	\end{equation}
	for some fixed $\epsilon,\delta > 0$ and for sufficiently large $J$. Then for any $p\in\mathbb{N}$
	\begin{equation}
		\label{eq:mainthm}
		\int_{T_{\rm{gap}}}^\infty dh\, \varphi_{h_*,\lambda_J}(h) \left[\rho_J(J + 2h)-\rho_{J,\rm{vac}}(J + 2h)\right] = \frac{1}{\sqrt{2J}} e^{4\pi\sqrt{AJ}} o\left(J^{-p}\right)\,,
	\end{equation}
	as $J \to \infty$. Here, $\varphi_{h_*,\lambda}$ is defined in \eqref{def:phirescale}, $\rho_J$ is the spectral density, and $\rho_{J,\rm{vac}}$ is defined in \eqref{rhoJvacdefn}. The right-hand side is uniform for $h_*$ in any bounded interval.
\end{theorem}

This theorem essentially states that the \emph{relative} error decays faster than any power law. Indeed, the contribution from $\rho_{J,\text{vac}}(J + 2h)$ behaves like $e^{4\pi \sqrt{AJ}}/ \sqrt{2J}$ by the previous proposition, although the prefactor is non-zero only if $h_* \geq A$. 

\begin{proof}
	It suffices to prove the case $\varphi \in \mathcal{D}([-1,1])$, since the general case $\varphi \in \mathcal{D}([-R,R])$ can be reduced to this by a simple substitution: set $\tilde{\varphi} := \varphi_{0,R^{-1}}$ and replace $\delta \rightarrow R\delta$. One can verify that 
	\begin{equation}
		\begin{split}
			\tilde{\varphi}_{h_*, R\lambda_J}(h) = \varphi_{h_*, \lambda_J}(h)\,.
		\end{split}
	\end{equation}
	Therefore, in the proof below, we restrict to the case $\varphi \in \mathcal{D}([-1,1])$ without loss of generality.
		
	At this stage we can completely forget about the partition function. We simply use that, for arbitrary $\beta > 0$,
	\begin{equation}
		\label{eq:rhoagainsttest}
		\begin{split}
		\int_{T_\text{gap}}^\infty dh \, \varphi(h) \rho_J(J + 2h) &= \int \frac{ds}{2\pi} \left[\int dh\, \varphi(h) e^{(\beta + i s)(h-A)} \right]F(\beta + is, J)\,,
		\end{split}
	\end{equation}
	and so the left-hand side of equation \eqref{eq:mainthm} is given by
	\begin{equation}
		\label{integral:phinonvac}
		\int \frac{ds}{2\pi} \hat \varphi_{h_*,\lambda}(s - i \beta) e^{- A (\beta+ i s)} F_\text{non-vac}(\beta + i s, J)\,.
	\end{equation}
	Proposition \ref{prop:Fnonvaccomplex} provides the bounds:
	\begin{equation}
		\label{Fnonvacboundssummary}
		\begin{split}
			\abs{F_{\rm{non-vac}}\left(\beta+ is,J\right)}\leqslant 
			&\,\sqrt{J}\,e^{4\pi \sqrt{AJ} \left(1 - \tau\right)} C'_\beta e^{\frac{2A}{\beta} s^2},\\
			\abs{F_{\rm{non-vac}}\left(\beta+ is,J\right)}
			\leq &\,\frac{1}{\sqrt{2J}} e^{4 \pi \sqrt{AJ}} C''_\beta
		\end{split}
	\end{equation}
	where $C_\beta', C_\beta''$ are independent of $s$ and $J$, and $\tau = \min \{3/8, T_\text{gap}/(2A)\}$. Notice also that for any $N\in\mathbb{N}$, $\hat{\varphi}(s-i\beta)$ has the following upper bound \cite{hormander2015analysis}
	\begin{equation}\label{varphi:FTbound}
		|\hat \varphi(s - i \beta)| \leq \frac{B_N e^{\beta}}{(1 + |s|)^N}\,,\quad{\rm for\ }s\in\mathbb{R}\ {\rm and\ }\beta\geqslant0.
	\end{equation}
	Here, $B_N$ is finite and depends only on $\varphi$ and $N$.\footnote{For any $\varphi\in\mathcal{D}([-1,1])$ we have, by integration-by-part $N$ times, 
	$$\abs{(s-i\beta)^N\hat{\varphi}(s-i\beta)}\equiv\abs{\int_{-1}^{1} dh \, \left[\partial^N\varphi(h)\right] e^{(\beta+i s) h}}\leqslant \abs{\abs{\partial^N\varphi}}_{L^1}e^{\abs{\beta}}.$$
	So $\abs{\hat{\varphi}(s-i\beta)}\leqslant \abs{s}^{-N}\abs{\abs{\partial^N\varphi}}_{L^1}e^{\abs{\beta}}$ for any $N\in\mathbb{N}$. The bound \eqref{varphi:FTbound} then follows by taking $B_N=2^N\max\left\{\abs{\abs{\varphi}}_{L^1},\,\abs{\abs{\partial^N\varphi}}_{L^1}\right\}$.} Then by \eqref{FT:rescaledphi}, we have
	\begin{equation}\label{bound:FTrescaledphi}
		\begin{split}
			\abs{\hat{\varphi}_{h_*,\lambda}(s-i\beta)}\leqslant \frac{B_N e^{(\lambda+h_*)\beta}}{(1+\lambda \abs{s})^{N}}\,.
		\end{split}
	\end{equation}
	Notice that this is where the dependence on $h_*$ enters the proof.

	We now split the integral in \eqref{integral:phinonvac} into two parts:
	\begin{equation}
		\begin{split}
			\int_{-\infty}^{+\infty}\frac{ds}{2\pi}=\int_{\abs{s}\leqslant \frac{1}{2}s_*}\frac{ds}{2\pi}+\int_{\abs{s}\geqslant \frac{1}{2}s_*}\frac{ds}{2\pi}
		\end{split}
	\end{equation}
	where $s_*$ is the transition point between the two bounds in equation \eqref{Fnonvacboundssummary}, which is given by
	\begin{equation}
		\begin{split}
			-4\pi\tau\sqrt{AJ}+\frac{2As_*^2}{\beta}+\log J=0\quad\Rightarrow\quad s_*=\sqrt{\frac{2\pi\beta\tau}{\sqrt{A}}}J^{1/4}\left(1-\frac{\log J}{4\pi\tau\sqrt{AJ}}\right)^{1/2}.
		\end{split}
	\end{equation}
	The integral over $\abs{s}\leqslant\frac{1}{2}s_*$ is bounded by
	\begin{equation}
		\begin{split}
			&B_N e^{(\lambda+h_*)\beta}C_\beta' J^{1/2}\,e^{4\pi\sqrt{AJ}(1-\tau)}\int_{\abs{s}\leqslant \frac{1}{2}s_*}\frac{ds}{2\pi}e^{\frac{2As^2}{\beta}} \\
			&\leqslant B_N e^{(\lambda+h_*)\beta}C_\beta' J^{1/2}\,e^{4\pi\sqrt{AJ}(1-\tau)}\left[ \int_{1\leqslant\abs{s}\leqslant\frac{1}{2}s_*}\frac{ds}{2\pi}\,s\,e^{\frac{2As^2}{\beta}}+	\int_{\abs{s}\leqslant1}\frac{ds}{2\pi}\,e^{\frac{2As^2}{\beta}}\right] \\
			&\leqslant \frac{1}{\pi}B_N e^{(\lambda+h_*)\beta}C_\beta' J^{1/2}\,e^{4\pi\sqrt{AJ}(1-\tau)}\left[\frac{\beta}{4A}e^{\frac{As_*^2}{2\beta}}+e^{\frac{2A}{\beta}}\left(1-\frac{\beta}{4A}\right)\right] \\
			&=O\left( J^{1/4}\,e^{4\pi\sqrt{AJ}(1-3\tau/4)+h_*\beta}\right).
		\end{split}
	\end{equation}
	The integral over $\abs{s}\geqslant \frac{1}{2}s_*$ is bounded by
	\begin{equation}
		\begin{split}
			&C''_\beta \frac{1}{\sqrt{2J}}\,e^{4\pi\sqrt{AJ}}\int_{\abs{s}\geqslant \frac{1}{2}s_*}\frac{ds}{2\pi} \frac{B_N\,  e^{(\lambda+h_*)\beta}}{(1+\lambda \abs{s})^{N}} \\
			&= B_N e^{(\lambda+h_*)\beta}\,C''_\beta \frac{1}{\sqrt{2J}}\,e^{4\pi\sqrt{AJ}} \frac{1}{\pi}\frac{1}{\lambda(N-1)}(1+\lambda s_*/2)^{-N+1}\\
			&=O\left( \frac{J^{-1/2} e^{4\pi\sqrt{AJ}+h_*\beta}}{\lambda \left(1+\lambda J^{1/4}\right)^{N-1}}\right) \\
		\end{split}
	\end{equation}
	This expression is strictly decreasing in $\lambda$ so the worst behavior occurs when $\lambda$ saturates the bound stated in the theorem, that is $\lambda =\delta/ J^{\frac{1}{4}  -\epsilon}$. In that case we get for the two parts of the integral:
	\begin{equation}\label{error:nonvac}
		\begin{split}
			\frac{1}{\sqrt{2J}}e^{4\pi\sqrt{AJ}}\times\begin{cases}
				O(J^{3/4}e^{-3\pi\tau\sqrt{AJ}+h_*\beta}), & \abs{s}\leqslant\frac{1}{2}s_*\,, \\
				O\left(J^{\frac{1}{4}-N\epsilon}e^{h_*\beta}\right), &\abs{s}\geqslant \frac{1}{2}s_*\,.
			\end{cases}
		\end{split}
	\end{equation}
	Now for any $\epsilon>0$ and $p\in\mathbb{N}$, we can choose sufficiently large $N$ such that $N\epsilon-1/4>p$, then the main statement follows, with our estimate for the subleading term coming from the part from $|s| \geq s_* / 2$.
	
	We finally note that our error estimates depend on $h_*$ only through the factors $e^{\beta h_*}$, which is uniformly bounded when $h_*$ lies in any finite interval.
\end{proof}
\begin{remark}
	The main statement of theorem \ref{theorem:error} also holds for a broader class of test functions, for example Gaussians like $\varphi(h) = e^{-a h^2}$. This is because the proof only relies on an upper bound for the Fourier transform of $\varphi$, namely equation \eqref{varphi:FTbound}, valid for $\beta$ in a finite interval. In fact, for Gaussian-type test functions, one can even obtain sharper bounds on the error term, which might be useful for other applications. In this work, however, we focus on compactly supported test functions, as our goal is to count states within a finite interval.
\end{remark}

\subsection{Consequences}
Theorem \ref{theorem:error} in particular tells us that when we smear the spin-$J$ spectral density of Virasoro primary states against a compactly supported smooth test function, its $1/J$ expansion is universal to all orders in $J$ and given by the vacuum term. However, we would like to remind the readers that it does not mean that the number of spin-$J$ Virasoro primary states in a finite twist window has a universal $1/J$ expansion. This is because the actual test function used for counting states is the indicator function which is not smooth. Nevertheless, theorem \ref{theorem:error} already has interesting physical implications, which we will describe below.

The first consequence of theorem \ref{theorem:error} is that the Cardy formula is valid when the size of the twist interval decays to zero as long as the process is slow enough:
\begin{corollary}\label{corr:count}
	Let $\delta>0$ and $\gamma\in[0,1/4)$ be fixed, let $\mathcal{N}(h_*,\delta,\gamma,J)$ be the number of spin-$J$ Virasoro primary states with $h\in(h_*-\delta J^{-\gamma},h_*+\delta J^{-\gamma})$. Then we have
	\begin{equation}
		\begin{split}
			\lim\limits_{J\rightarrow\infty}\left[\sqrt{2}J^{\frac{1}{2}+\gamma}e^{-4\pi\sqrt{AJ}}\mathcal{N}(h_*,\delta,\gamma,J)\right]=\begin{cases}
				\displaystyle\int_{h_*-\delta}^{h_*+\delta} dh\, \rho_{\rm c}(h), & \gamma = 0 , \\[2.5ex]
				\displaystyle  2\delta\rho_{\rm c}(h_*) & 0<\gamma<1/4. \\
			\end{cases}
		\end{split}
	\end{equation}
\end{corollary}
\begin{proof}
	Consider two types of compactly supported smooth test functions $\varphi^\pm$ satisfying
	\begin{equation}\label{def:phipm}
		\begin{split}
			\varphi^-(x)<\theta_{(-1,1)}(x)<\varphi^+(x),
		\end{split}
	\end{equation}
	where $\theta_{(-1,1)}(x)$ is the indicator function of $(-1,1)$. See figure \ref{fig:phipm} for an example.
	\begin{figure}[t]
		\centering
		\includegraphics[width=0.7\textwidth]{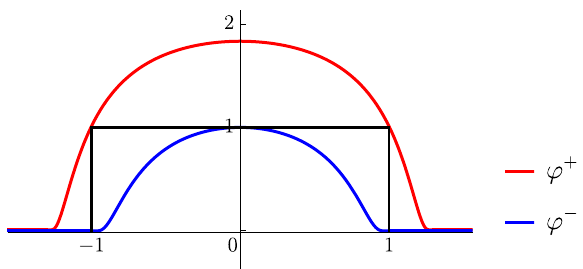} 
		\caption{An example of compactly supported smooth functions $\varphi^+$ (red) and $\varphi^-$ (blue). They are constructed from bump functions of the form $\varphi^+(x)=\exp\left(-\frac{b_+}{a_+^2 - x^2} + \frac{b_+}{a_+^2 - 1}\right)$ (for $\abs{x}\leqslant a_+$) and $\varphi^-(x)=\exp\left(-\frac{b_-}{a_-^2 - x^2} + \frac{b_-}{a_-^2}\right)$ (for $\abs{x}\leqslant a_-$), with appropriately chosen parameters.}
		\label{fig:phipm}
	\end{figure}
	Define
	\begin{equation}
	\begin{split}
		\varphi^{\pm}_{h_*,\delta J^{-\gamma}}(h):=\delta^{-1}J^{\gamma}\varphi^{\pm}(\delta^{-1}J^{\gamma}(h-h_*))\,
	\end{split}
\end{equation}
	according to \eqref{def:phirescale}.
	
	Since the spectral density is positive, we have
	\begin{equation}
		\begin{split}
			2\delta J^{-\gamma}\int dh\,\rho_J(J+2h)\,\varphi_{h_*,\delta J^{-\gamma}}^-(h)\leqslant\mathcal{N}(h_*,\delta,\gamma,J)\leqslant 2\delta J^{-\gamma}\int dh\,\rho_J(J+2h)\,\varphi_{h_*,\delta J^{-\gamma}}^+(h).
		\end{split}
	\end{equation}
	Let us first focus on the upper bound. By theorem \ref{theorem:error}  and proposition \ref{prop:rhovaclimituniform}, we obtain
	\begin{equation}\label{limsupN}
		\begin{split}
			\limsup\limits_{J\rightarrow\infty}\left[\sqrt{2}J^{\frac{1}{2}+\gamma}e^{-4\pi\sqrt{AJ}}\mathcal{N}(h_*,\delta,\gamma,J)\right]\leqslant \begin{cases}
				\displaystyle\delta\int_{h_*-\delta}^{h_*+\delta} dh\, \rho_{\rm c}(h)\,\varphi^+_{h_*,\delta}(h), & \gamma = 0 , \\[2.5ex]
				\displaystyle  \delta\rho_{\rm c}(h_*)\int dx\,\varphi^+(x) & 0<\gamma<1/4. \\
			\end{cases}
		\end{split}
	\end{equation}
	for any $\varphi^+$ satisfying the above conditions. Notice that proposition \ref{prop:rhovaclimituniform} does not apply for $\gamma<0$ because the support of the test function stops being bounded, so we lose uniformity.
	
	The left-hand side of \eqref{limsupN} does not depend on $\varphi^+$, so we can optimize the upper bound by taking the infimum over the allowed choices of $\varphi^+$. Recall that $\varphi^+$ is required to be greater than the indicator function $\theta_{(-1,1)}$, but may be taken arbitrarily close to it. Hence, we have
	\begin{equation}
		\begin{split}
			&\inf\limits_{\varphi^+ > \theta_{(-1,1)}} \int_{h_* - \delta}^{h_* + \delta} dh\, \rho_{\rm c}(h)\, \varphi^+_{h_*, \delta}(h) = \delta^{-1} \int_{h_* - \delta}^{h_* + \delta} dh\, \rho_{\rm c}(h), \\
			&\inf\limits_{\varphi^+ > \theta_{(-1,1)}} \int dx\, \varphi^+(x) = 2.
		\end{split}
	\end{equation}
	Therefore, taking the infimum on the right-hand side of \eqref{limsupN} yields
	\begin{equation}
		\begin{split}
			\limsup\limits_{J \rightarrow \infty} \left[\sqrt{2}\, J^{\frac{1}{2} + \gamma} e^{-4\pi\sqrt{A J}} \mathcal{N}(h_*, \delta, \gamma, J)\right] \leqslant \begin{cases}
				\displaystyle \int_{h_* - \delta}^{h_* + \delta} dh\, \rho_{\rm c}(h), & \gamma = 0, \\[2.5ex]
				\displaystyle 2\delta\, \rho_{\rm c}(h_*), & 0 < \gamma < 1/4.
			\end{cases}
		\end{split}
	\end{equation}

	The analogous argument applied to $\varphi^-$ gives the corresponding lower bound on the liminf, which is the same as the upper bound on the limsup. This proves the corollary.
\end{proof}

The second consequence of theorem \ref{theorem:error} is about the level spacing, which is defined to be the maximal difference between two neighboring states in the spectrum. To be precise, let us fix a spin $J$. Notice that ${\rm supp}(\rho_J)$ is a closed set, so $(a,b)\backslash {\rm supp}(\rho_J)$ is a disjoint union of open intervals. The maximal level spacing $\ell_J(a,b)$ is then defined by the maximal size among these small intervals, or zero if ${\rm supp}(\rho_J)$ covers the whole interval $(a,b)$.

For example, when the spectrum within interval $(a,b)$ is finite, we can order the spectrum:
$$a<h_1<h_2<h_3<\ldots<h_N< b.$$
Then the maximal level spacing within this interval is given by
\begin{equation}
	\ell_J(a,b)=\max\limits_{0\leqslant i\leqslant N+1}\abs{h_{i+1}-h_i},\quad {\rm where\ } h_0=a,\ h_{N+1}=b.
\end{equation}
\begin{corollary}\label{corr:spacing}
	Let $\ell_J(a,b)$ be the maximal level spacing of spin-$J$ Virasoro primary states with $h\in (a,b)\subset(A,\infty)$. Then we have
	\begin{equation}
		\begin{split}
			\lim\limits_{J\rightarrow\infty}J^{\frac{1}{4}-\epsilon}\ell_J(a,b)=0,\quad\forall \epsilon>0.
		\end{split}
	\end{equation}
\end{corollary}

\begin{proof}
It suffices to prove the statement for the sequence of intervals $(A+1/J,b)$ because
\begin{equation}
	\ell_J(a,b)\leqslant\ell_J(A,b)\leqslant\ell_J(A+1/J,b)+1/J,\quad\forall a\geqslant A,
\end{equation}
which implies
\begin{equation}
	\limsup\limits_{J\rightarrow\infty}J^{\frac{1}{4}-\epsilon}\ell_J(a,b)\leqslant\limsup\limits_{J\rightarrow\infty}J^{\frac{1}{4}-\epsilon}\ell_J(A+1/J,b),\quad \forall a\geqslant A,\ \epsilon>0.
\end{equation}
Pick a test function $\varphi(h)$ in $\mathcal{D}([-1,1])$ which is non-negative, $\varphi(h)\geqslant 0$, and not identically zero, so $\hat \varphi(0)>0$. If the corollary would not hold then
\begin{equation}
	\label{corollaryaux}
	\liminf_{J \to \infty} \inf_{h_* \in (A+1/J,b)} \int_{T_\text{gap}}^\infty\ dh\ \varphi_{h_{*}, \delta J^{-1/4+\epsilon}}(h) \rho_{J}(J + 2h) = 0
\end{equation}
since we can then pick a sequence $J_n \to \infty$ and $h_{*,n} \in (A+1/J,b)$ such that there are no operators in an interval of width $2 \delta J^{-1/4 + \epsilon}$ around $h_{*,n}$. In actuality, however, theorem \ref{theorem:error} and proposition \ref{prop:rhovaclimituniform} established that
\begin{equation}
	\int_{T_\text{gap}}^\infty\ dh\ \varphi_{h_{*}, \delta J^{-1/4+\epsilon}}(h) \rho_{J}(J + 2h) \underset{J \to \infty}\sim \hat{\varphi}(0) \rho_{\rm c}(h_*) \frac{e^{4 \pi \sqrt{AJ}}}{2\sqrt{2J}}\,.
\end{equation}
Even better, we know that this behavior holds uniformly for $h_*$ in a finite interval like $(A,b)$, so for sufficiently large $J$ we have
\begin{equation}
	\int_{T_\text{gap}}^\infty\ dh\ \varphi_{h_{*}, \delta J^{-1/4+\epsilon}}(h) \rho_{J}(J + 2h) \geqslant \frac{1}{2}\hat{\varphi}(0) \rho_{\rm c}(h_*) \frac{e^{4 \pi \sqrt{AJ}}}{2\sqrt{2J}}\,,
\end{equation}
uniformly for $h_* \in (A,b)$. Now since we are taking the infimum among $h_*\in(A+1/J,b)$, by the fact that $\rho_{\rm c}(h_*)$ is monotonically increasing in $h_*$, it suffices to consider the value of $\rho_{\rm c}(h_*)$ at $h_*=A+1/J$. 

Since $\rho_{\rm c}(h)=8\pi^2\sqrt{2(h-A)}+O\left((h-A)^{3/2}\right)$ near $h=A$, we have
\begin{equation}
	\begin{split}
		\liminf_{J \to \infty}\inf_{h_* \in [A+1/J,b]}J^{1/2}\rho_{\rm c}(h_*)>0,
	\end{split}
\end{equation}
for fixed $b$. Therefore,
\begin{equation}
	\begin{split}
		\liminf_{J \to \infty} \inf_{h_* \in [A+1/J,b]}J e^{-4\pi\sqrt{AJ}}\int_{T_\text{gap}}^\infty \varphi_{h_{*}, \delta J^{-1/4+\epsilon}}(h) \rho_{J}(J + 2h)>0\,,
	\end{split}
\end{equation}
since the error from the non-vacuum term is suppressed by a higher power of $1/J$ according to theorem \ref{theorem:error}. Therefore, the left-hand side of equation \eqref{corollaryaux} is actually infinity.
\end{proof}

\section{Outlook}
In this paper, we have established rigorous results about the density of Virasoro primaries with large spin in two-dimensional unitary CFTs with $c>1$, having a normalizable vacuum and a twist gap in the spectrum of Virasoro primaries.  The number of Virasoro primaries with twist greater than or equal to $(c-1)/12$ has a Cardy-like growth, \emph{i.e.}  $(2J)^{-1/2}e^{4\pi\sqrt{AJ}}$ for large $J$ while  the number of Virasoro primaries strictly below $(c-1)/12$ has a sub-Cardy growth, \emph{i.e.} $o\left(J^{-1/2}e^{4\pi\sqrt{AJ}}\right)$.  Furthermore,  we have proved the maximal level spacing of Virasoro primary states with spin $J$,  and twist lying in some bounded interval contained in $[(c-1)/12,\infty)$,  goes to $0$ at least as fast as $J^{-1/4+\epsilon}$ as $J\to\infty$.  

Our bounds followed from simple estimates on the universal and non-universal terms in the modular invariant partition function. In fact, from our proofs it is clear that we can also provide estimates for the operator density any \emph{finite} spin $J$, with errors expressed only in terms of $c$, $T_\text{gap}$ and $\tilde Z(2\pi,2\pi)$, which corresponds to the value of the partition function at the modular invariant point $\tau = i$. It would be very interesting to work out the details. One might even use the numerical modular bootstrap of \cite{Hellerman:2009bu}, see also \cite{Collier:2016cls,Afkhami-Jeddi:2019zci,Hartman:2019pcd,Benjamin:2022pnx}, to constrain $\tilde Z(2\pi,2\pi)$ in terms of $c$ and $T_\text{gap}$ and thus eliminate it from our error estimates. It would also be nice to take into account subleading terms or perhaps the other modular images of the vacuum term, leading possibly to a Rademacher-type expansion \cite{rademacher1938fourier,rademacher1943expansion,Eberhardt:2023xck,Baccianti:2025gll}.

Another natural direction for future research is to extend these results to CFTs with discrete global symmetry following the approach of \cite{Pal:2020wwd}, see also \cite{Harlow:2021trr,Kang:2022orq,Lin:2022dhv}.  One expects that the results should hold independently within each symmetry sector, labeled by irreducible representations. The numerical modular bootstrap in the presence of discrete global symmetries was studied in \cite{Lin:2021udi,Lin:2023uvm}. For CFTs with continuous symmetries,  the Virasoro symmetry often gets extended to a bigger chiral algebra. In this case,  one needs to identify the largest chiral algebra and impose a twist gap with respect to its primaries. 

One could further hope to study torus two-point function and genus-two partition functions using similar techniques.   However, the relevant conformal blocks in these cases are more intricate than the simple exponential form $e^{-\beta_L h-\beta_R \bar h}$.   A particularly interesting analogue of the higher-dimensional lightcone bootstrap would be the analysis of sphere four-point functions. While in higher dimensions, the large-spin limit leads to the emergence of mean field theory \cite{Fitzpatrick:2012yx,Komargodski:2012ek,Pal:2022vqc,vanRees:2024xkb}, in two dimensions, the corresponding universal theory is known as Virasoro mean field theory (VMFT) \cite{Collier:2018exn}. Establishing a rigorous foundation for VMFT using the techniques of this paper would be a significant step forward.

CFTs with large central charge and appropriate sparseness condition exhibit universal features as shown in \cite{Hartman:2014oaa} and \cite{Dey:2024nje}.  These features are the hallmarks of semiclassical Einstein gravity and intimately connected with black hole microstate counting \cite{Strominger:1996sh,Strominger:1997eq,Sen:2007qy,Sen:2012dw,Hartman:2014oaa}.  It would be intriguing to extend the results of the present paper to the large central charge regime,  and explore their implications for highly rotating BTZ black holes in AdS$_3$.   Finally,  we expect that the techniques of this paper may prove useful to generalizations of the Cardy formula appearing in, for example, \cite{Verlinde:2000wg,Kutasov:2000td,Shaghoulian:2015kta,Kraus:2016nwo,Das:2017vej, Das:2017cnv,Cardy:2017qhl,Brehm:2018ipf,Collier:2019weq,Pal:2020wwd,Harlow:2021trr,Belin:2021ryy,Anous:2021caj,Lin:2022dhv,Benjamin:2023qsc,Benjamin:2024kdg}.

\section*{Acknowledgments}
We are grateful to Lorenz Eberhardt, Dalimil Mazáč, Eric Perlmutter and Ioannis Tsiares for discussions. SP is supported by the U.S. Department of Energy, Office of Science, Office of High Energy Physics, under Award Number DE-SC0011632, and by the Walter Burke Institute for Theoretical Physics. JQ is supported by Simons Foundation grant 994310 (Simons Collaboration on Confinement and QCD Strings). BvR acknowledges funding from the European Union (ERC “QFTinAdS”, project number 101087025).

\bibliography{lightcone}

\providecommand{\href}[2]{#2}\begingroup\raggedright\begin{thebibliography}{10}

\bibitem{Antunes:2022vtb}
A.~Antunes and C.~Behan, ``{Coupled Minimal Conformal Field Theory Models
  Revisited},'' \href{http://dx.doi.org/10.1103/PhysRevLett.130.071602}{{\em
  Phys. Rev. Lett.} {\bfseries 130} no.~7, (2023) 071602},
  \href{http://arxiv.org/abs/2211.16503}{{\ttfamily arXiv:2211.16503
  [hep-th]}}.

\bibitem{Antunes:2024mfb}
A.~Antunes and C.~Behan, ``{Coupled minimal models revisited II: Constraints
  from permutation symmetry},''
  \href{http://arxiv.org/abs/2412.21107}{{\ttfamily arXiv:2412.21107
  [hep-th]}}.

\bibitem{Pal:2023cgk}
S.~Pal and J.~Qiao, ``{Lightcone Modular Bootstrap and Tauberian Theory: A
  Cardy-like Formula for Near-extremal Black Holes},''
  \href{http://dx.doi.org/10.1007/s00023-024-01441-2}{{\em Annales Henri
  Poincar{\'e}} (2024) }, \href{http://arxiv.org/abs/2307.02587}{{\ttfamily
  arXiv:2307.02587 [hep-th]}}.

\bibitem{Collier:2016cls}
S.~Collier, Y.-H. Lin, and X.~Yin, ``{Modular Bootstrap Revisited},''
  \href{http://dx.doi.org/10.1007/JHEP09(2018)061}{{\em JHEP} {\bfseries 09}
  (2018) 061}, \href{http://arxiv.org/abs/1608.06241}{{\ttfamily
  arXiv:1608.06241 [hep-th]}}.

\bibitem{Afkhami-Jeddi:2017idc}
N.~Afkhami-Jeddi, K.~Colville, T.~Hartman, A.~Maloney, and E.~Perlmutter,
  ``{Constraints on higher spin CFT$_{2}$},''
  \href{http://dx.doi.org/10.1007/JHEP05(2018)092}{{\em JHEP} {\bfseries 05}
  (2018) 092}, \href{http://arxiv.org/abs/1707.07717}{{\ttfamily
  arXiv:1707.07717 [hep-th]}}.

\bibitem{Benjamin:2019stq}
N.~Benjamin, H.~Ooguri, S.-H. Shao, and Y.~Wang, ``{Light-cone modular
  bootstrap and pure gravity},''
  \href{http://dx.doi.org/10.1103/PhysRevD.100.066029}{{\em Phys. Rev. D}
  {\bfseries 100} no.~6, (2019) 066029},
  \href{http://arxiv.org/abs/1906.04184}{{\ttfamily arXiv:1906.04184
  [hep-th]}}.

\bibitem{Cardy:1986ie}
J.~L. Cardy, ``{Operator Content of Two-Dimensional Conformally Invariant
  Theories},'' \href{http://dx.doi.org/10.1016/0550-3213(86)90552-3}{{\em Nucl.
  Phys. B} {\bfseries 270} (1986) 186--204}.

\bibitem{korevaar2004tauberian}
J.~Korevaar, \href{http://dx.doi.org/10.1007/978-3-662-10225-1}{{\em Tauberian
  theory: a century of developments}}.
\newblock Springer, 2004.

\bibitem{Pappadopulo:2012jk}
D.~Pappadopulo, S.~Rychkov, J.~Espin, and R.~Rattazzi, ``{OPE Convergence in
  Conformal Field Theory},''
  \href{http://dx.doi.org/10.1103/PhysRevD.86.105043}{{\em Phys. Rev. D}
  {\bfseries 86} (2012) 105043},
  \href{http://arxiv.org/abs/1208.6449}{{\ttfamily arXiv:1208.6449 [hep-th]}}.

\bibitem{Hellerman:2009bu}
S.~Hellerman, ``{A Universal Inequality for CFT and Quantum Gravity},''
  \href{http://dx.doi.org/10.1007/JHEP08(2011)130}{{\em JHEP} {\bfseries 08}
  (2011) 130}, \href{http://arxiv.org/abs/0902.2790}{{\ttfamily arXiv:0902.2790
  [hep-th]}}.

\bibitem{Zamolodchikov:1987avt}
A.~B. Zamolodchikov, ``{Conformal symmetry in two-dimensional space: Recursion
  representation of conformal block},''
  \href{http://dx.doi.org/10.1007/BF01022967}{{\em Theor. Math. Phys.}
  {\bfseries 73} no.~1, (1987) 1088--1093}.

\bibitem{Maldacena:2015iua}
J.~Maldacena, D.~Simmons-Duffin, and A.~Zhiboedov, ``{Looking for a bulk
  point},'' \href{http://dx.doi.org/10.1007/JHEP01(2017)013}{{\em JHEP}
  {\bfseries 01} (2017) 013}, \href{http://arxiv.org/abs/1509.03612}{{\ttfamily
  arXiv:1509.03612 [hep-th]}}.

\bibitem{Maloney:2016kee}
A.~Maloney, H.~Maxfield, and G.~S. Ng, ``{A conformal block Farey tail},''
  \href{http://dx.doi.org/10.1007/JHEP06(2017)117}{{\em JHEP} {\bfseries 06}
  (2017) 117}, \href{http://arxiv.org/abs/1609.02165}{{\ttfamily
  arXiv:1609.02165 [hep-th]}}.

\bibitem{Das:2017cnv}
D.~Das, S.~Datta, and S.~Pal, ``{Universal asymptotics of three-point
  coefficients from elliptic representation of Virasoro blocks},''
  \href{http://dx.doi.org/10.1103/PhysRevD.98.101901}{{\em Phys. Rev. D}
  {\bfseries 98} no.~10, (2018) 101901},
  \href{http://arxiv.org/abs/1712.01842}{{\ttfamily arXiv:1712.01842
  [hep-th]}}.

\bibitem{Hartman:2019pcd}
T.~Hartman, D.~Maz\'a\v{c}, and L.~Rastelli, ``{Sphere Packing and Quantum
  Gravity},'' \href{http://dx.doi.org/10.1007/JHEP12(2019)048}{{\em JHEP}
  {\bfseries 12} (2019) 048}, \href{http://arxiv.org/abs/1905.01319}{{\ttfamily
  arXiv:1905.01319 [hep-th]}}.

\bibitem{Pal:2022vqc}
S.~Pal, J.~Qiao, and S.~Rychkov, ``{Twist accumulation in conformal field
  theory. A rigorous approach to the lightcone bootstrap},''
  \href{http://dx.doi.org/10.1007/s00220-023-04767-w}{{\em Commun. Math. Phys.}
  (2023) }, \href{http://arxiv.org/abs/2212.04893}{{\ttfamily arXiv:2212.04893
  [hep-th]}}.

\bibitem{Qiao:2017xif}
J.~Qiao and S.~Rychkov, ``{A tauberian theorem for the conformal bootstrap},''
  \href{http://dx.doi.org/10.1007/JHEP12(2017)119}{{\em JHEP} {\bfseries 12}
  (2017) 119}, \href{http://arxiv.org/abs/1709.00008}{{\ttfamily
  arXiv:1709.00008 [hep-th]}}.

\bibitem{Mukhametzhanov:2018zja}
B.~Mukhametzhanov and A.~Zhiboedov, ``{Analytic Euclidean Bootstrap},''
  \href{http://dx.doi.org/10.1007/JHEP10(2019)270}{{\em JHEP} {\bfseries 10}
  (2019) 270}, \href{http://arxiv.org/abs/1808.03212}{{\ttfamily
  arXiv:1808.03212 [hep-th]}}.

\bibitem{Mukhametzhanov:2019pzy}
B.~Mukhametzhanov and A.~Zhiboedov, ``{Modular invariance, tauberian theorems
  and microcanonical entropy},''
  \href{http://dx.doi.org/10.1007/JHEP10(2019)261}{{\em JHEP} {\bfseries 10}
  (2019) 261}, \href{http://arxiv.org/abs/1904.06359}{{\ttfamily
  arXiv:1904.06359 [hep-th]}}.

\bibitem{Ganguly:2019ksp}
S.~Ganguly and S.~Pal, ``{Bounds on the density of states and the spectral gap
  in CFT$_{2}$},'' \href{http://dx.doi.org/10.1103/PhysRevD.101.106022}{{\em
  Phys. Rev. D} {\bfseries 101} no.~10, (2020) 106022},
  \href{http://arxiv.org/abs/1905.12636}{{\ttfamily arXiv:1905.12636
  [hep-th]}}.

\bibitem{Mukhametzhanov:2020swe}
B.~Mukhametzhanov and S.~Pal, ``{Beurling-Selberg Extremization and Modular
  Bootstrap at High Energies},''
  \href{http://dx.doi.org/10.21468/SciPostPhys.8.6.088}{{\em SciPost Phys.}
  {\bfseries 8} no.~6, (2020) 088},
  \href{http://arxiv.org/abs/2003.14316}{{\ttfamily arXiv:2003.14316
  [hep-th]}}.

\bibitem{Fitzpatrick:2012yx}
A.~L. Fitzpatrick, J.~Kaplan, D.~Poland, and D.~Simmons-Duffin, ``{The Analytic
  Bootstrap and AdS Superhorizon Locality},''
  \href{http://dx.doi.org/10.1007/JHEP12(2013)004}{{\em JHEP} {\bfseries 12}
  (2013) 004}, \href{http://arxiv.org/abs/1212.3616}{{\ttfamily arXiv:1212.3616
  [hep-th]}}.

\bibitem{Komargodski:2012ek}
Z.~Komargodski and A.~Zhiboedov, ``{Convexity and Liberation at Large Spin},''
  \href{http://dx.doi.org/10.1007/JHEP11(2013)140}{{\em JHEP} {\bfseries 11}
  (2013) 140}, \href{http://arxiv.org/abs/1212.4103}{{\ttfamily arXiv:1212.4103
  [hep-th]}}.

\bibitem{vanRees:2024xkb}
B.~C. van Rees, ``{Theorems for the Lightcone Bootstrap},''
  \href{http://arxiv.org/abs/2412.06907}{{\ttfamily arXiv:2412.06907
  [hep-th]}}.

\bibitem{Hartman:2014oaa}
T.~Hartman, C.~A. Keller, and B.~Stoica, ``{Universal Spectrum of 2d Conformal
  Field Theory in the Large c Limit},''
  \href{http://dx.doi.org/10.1007/JHEP09(2014)118}{{\em JHEP} {\bfseries 09}
  (2014) 118}, \href{http://arxiv.org/abs/1405.5137}{{\ttfamily arXiv:1405.5137
  [hep-th]}}.

\bibitem{Berry:1977levelclustering}
M.~V. {Berry} and M.~{Tabor}, ``{Level Clustering in the Regular Spectrum},''
  \href{http://dx.doi.org/10.1098/rspa.1977.0140}{{\em Proceedings of the Royal
  Society of London Series A} {\bfseries 356} no.~1686, (Sep, 1977) 375--394}.

\bibitem{Bohigas:1983er}
O.~Bohigas, M.~J. Giannoni, and C.~Schmit, ``{Characterization of chaotic
  quantum spectra and universality of level fluctuation laws},''
  \href{http://dx.doi.org/10.1103/PhysRevLett.52.1}{{\em Phys. Rev. Lett.}
  {\bfseries 52} (1984) 1--4}.

\bibitem{Wigner:1951stat}
E.~P. Wigner, ``{On the statistical distribution of the widths and spacings of
  nuclear resonance levels},''
  \href{http://dx.doi.org/10.1017/S0305004100027237}{{\em Proceedings of the
  Cambridge Philosophical Society} {\bfseries 47} no.~4, (1951) 790--798}.

\bibitem{PorterThomas:1956}
C.~E. {Porter} and R.~G. {Thomas}, ``{Fluctuations of Nuclear Reaction
  Widths},'' \href{http://dx.doi.org/10.1103/PhysRev.104.483}{{\em Physical
  Review} {\bfseries 104} no.~2, (Oct., 1956) 483--491}.

\bibitem{Mehta:1960stat}
M.~L. {Mehta}, ``{On the statistical properties of the level-spacings in
  nuclear spectra},''
  \href{http://dx.doi.org/10.1016/0029-5582(60)90413-2}{{\em Nuclear Physics}
  {\bfseries 18} (Sept., 1960) 395--419}.

\bibitem{Mehta:1967book}
M.~L. Mehta, {\em Random Matrices and the Statistical Theory of Energy Levels}.
\newblock Academic Press, New York, 1967.

\bibitem{montgomery1973pair}
H.~L. Montgomery, ``The pair correlation of zeros of the zeta function,'' in
  {\em Proc. Symp. Pure Math}, vol.~24, p.~1.
\newblock 1973.

\bibitem{odlyzko1987distribution}
A.~M. Odlyzko, ``On the distribution of spacings between zeros of the zeta
  function,'' {\em Mathematics of Computation} {\bfseries 48} no.~177, (1987)
  273--308.

\bibitem{marklof2006arithmetic}
J.~Marklof, ``Arithmetic quantum chaos,'' {\em Encyclopedia of mathematical
  physics} {\bfseries 1} (2006) 212--220.

\bibitem{Kravchuk:2021akc}
P.~Kravchuk, D.~Mazac, and S.~Pal, ``{Automorphic spectra and the conformal
  bootstrap},'' \href{http://dx.doi.org/10.1090/cams/26}{{\em Commun. Am. Math.
  Soc.} {\bfseries 4} no.~1, (2024) 1--63},
  \href{http://arxiv.org/abs/2111.12716}{{\ttfamily arXiv:2111.12716
  [hep-th]}}.

\bibitem{Cotler:2016fpe}
J.~S. Cotler, G.~Gur-Ari, M.~Hanada, J.~Polchinski, P.~Saad, S.~H. Shenker,
  D.~Stanford, A.~Streicher, and M.~Tezuka, ``{Black Holes and Random
  Matrices},'' \href{http://dx.doi.org/10.1007/JHEP05(2017)118}{{\em JHEP}
  {\bfseries 05} (2017) 118}, \href{http://arxiv.org/abs/1611.04650}{{\ttfamily
  arXiv:1611.04650 [hep-th]}}. [Erratum: JHEP 09, 002 (2018)].

\bibitem{Saad:2019lba}
P.~Saad, S.~H. Shenker, and D.~Stanford, ``{JT gravity as a matrix integral},''
  \href{http://arxiv.org/abs/1903.11115}{{\ttfamily arXiv:1903.11115
  [hep-th]}}.

\bibitem{Srednicki:1994mfb}
M.~Srednicki, ``{Chaos and Quantum Thermalization},''
  \href{http://dx.doi.org/10.1103/PhysRevE.50.888}{{\em Phys. Rev. E}
  {\bfseries 50} (3, 1994) },
  \href{http://arxiv.org/abs/cond-mat/9403051}{{\ttfamily
  arXiv:cond-mat/9403051}}.

\bibitem{Kraus:2016nwo}
P.~Kraus and A.~Maloney, ``{A cardy formula for three-point coefficients or how
  the black hole got its spots},''
  \href{http://dx.doi.org/10.1007/JHEP05(2017)160}{{\em JHEP} {\bfseries 05}
  (2017) 160}, \href{http://arxiv.org/abs/1608.03284}{{\ttfamily
  arXiv:1608.03284 [hep-th]}}.

\bibitem{Das:2017vej}
D.~Das, S.~Datta, and S.~Pal, ``{Charged structure constants from
  modularity},'' \href{http://dx.doi.org/10.1007/JHEP11(2017)183}{{\em JHEP}
  {\bfseries 11} (2017) 183}, \href{http://arxiv.org/abs/1706.04612}{{\ttfamily
  arXiv:1706.04612 [hep-th]}}.

\bibitem{Basu:2017kzo}
P.~Basu, D.~Das, S.~Datta, and S.~Pal, ``{Thermality of eigenstates in
  conformal field theories},''
  \href{http://dx.doi.org/10.1103/PhysRevE.96.022149}{{\em Phys. Rev. E}
  {\bfseries 96} no.~2, (2017) 022149},
  \href{http://arxiv.org/abs/1705.03001}{{\ttfamily arXiv:1705.03001
  [hep-th]}}.

\bibitem{Brehm:2018ipf}
E.~M. Brehm, D.~Das, and S.~Datta, ``{Probing thermality beyond the
  diagonal},'' \href{http://dx.doi.org/10.1103/PhysRevD.98.126015}{{\em Phys.
  Rev. D} {\bfseries 98} no.~12, (2018) 126015},
  \href{http://arxiv.org/abs/1804.07924}{{\ttfamily arXiv:1804.07924
  [hep-th]}}.

\bibitem{Romero-Bermudez:2018dim}
A.~Romero-Berm\'udez, P.~Sabella-Garnier, and K.~Schalm, ``{A Cardy formula for
  off-diagonal three-point coefficients; or, how the geometry behind the
  horizon gets disentangled},''
  \href{http://dx.doi.org/10.1007/JHEP09(2018)005}{{\em JHEP} {\bfseries 09}
  (2018) 005}, \href{http://arxiv.org/abs/1804.08899}{{\ttfamily
  arXiv:1804.08899 [hep-th]}}.

\bibitem{Hikida:2018khg}
Y.~Hikida, Y.~Kusuki, and T.~Takayanagi, ``{Eigenstate thermalization
  hypothesis and modular invariance of two-dimensional conformal field
  theories},'' \href{http://dx.doi.org/10.1103/PhysRevD.98.026003}{{\em Phys.
  Rev. D} {\bfseries 98} no.~2, (2018) 026003},
  \href{http://arxiv.org/abs/1804.09658}{{\ttfamily arXiv:1804.09658
  [hep-th]}}.

\bibitem{Collier:2019weq}
S.~Collier, A.~Maloney, H.~Maxfield, and I.~Tsiares, ``{Universal dynamics of
  heavy operators in CFT$_{2}$},''
  \href{http://dx.doi.org/10.1007/JHEP07(2020)074}{{\em JHEP} {\bfseries 07}
  (2020) 074}, \href{http://arxiv.org/abs/1912.00222}{{\ttfamily
  arXiv:1912.00222 [hep-th]}}.

\bibitem{Chen:2024lji}
L.~Chen, A.~Dymarsky, J.~Tian, and H.~Wang, ``{Subsystem entropy in 2d CFT and
  KdV ETH},'' \href{http://arxiv.org/abs/2409.19046}{{\ttfamily
  arXiv:2409.19046 [hep-th]}}.

\bibitem{Roberts:2014ifa}
D.~A. Roberts and D.~Stanford, ``{Two-dimensional conformal field theory and
  the butterfly effect},''
  \href{http://dx.doi.org/10.1103/PhysRevLett.115.131603}{{\em Phys. Rev.
  Lett.} {\bfseries 115} no.~13, (2015) 131603},
  \href{http://arxiv.org/abs/1412.5123}{{\ttfamily arXiv:1412.5123 [hep-th]}}.

\bibitem{Haehl:2018izb}
F.~M. Haehl and M.~Rozali, ``{Effective Field Theory for Chaotic CFTs},''
  \href{http://dx.doi.org/10.1007/JHEP10(2018)118}{{\em JHEP} {\bfseries 10}
  (2018) 118}, \href{http://arxiv.org/abs/1808.02898}{{\ttfamily
  arXiv:1808.02898 [hep-th]}}.

\bibitem{Anous:2020vtw}
T.~Anous and F.~M. Haehl, ``{On the Virasoro six-point identity block and
  chaos},'' \href{http://dx.doi.org/10.1007/JHEP08(2020)002}{{\em JHEP}
  {\bfseries 08} no.~08, (2020) 002},
  \href{http://arxiv.org/abs/2005.06440}{{\ttfamily arXiv:2005.06440
  [hep-th]}}.

\bibitem{Altland:2020ccq}
A.~Altland and J.~Sonner, ``{Late time physics of holographic quantum chaos},''
  \href{http://dx.doi.org/10.21468/SciPostPhys.11.2.034}{{\em SciPost Phys.}
  {\bfseries 11} (2021) 034}, \href{http://arxiv.org/abs/2008.02271}{{\ttfamily
  arXiv:2008.02271 [hep-th]}}.

\bibitem{Choi:2023mab}
C.~Choi, F.~M. Haehl, M.~Mezei, and G.~S\'arosi, ``{Effective description of
  sub-maximal chaos: stringy effects for SYK scrambling},''
  \href{http://arxiv.org/abs/2301.05698}{{\ttfamily arXiv:2301.05698
  [hep-th]}}.

\bibitem{Boruch:2025ilr}
J.~Boruch, G.~Di~Ubaldo, F.~M. Haehl, E.~Perlmutter, and M.~Rozali,
  ``{Modular-invariant random matrix theory and AdS${}_3$ wormholes},''
  \href{http://arxiv.org/abs/2503.00101}{{\ttfamily arXiv:2503.00101
  [hep-th]}}.

\bibitem{Haehl:2023mhf}
F.~M. Haehl, W.~Reeves, and M.~Rozali, ``{Euclidean wormholes in
  two-dimensional conformal field theories from quantum chaos and number
  theory},'' \href{http://dx.doi.org/10.1103/PhysRevD.108.L101902}{{\em Phys.
  Rev. D} {\bfseries 108} no.~10, (2023) L101902},
  \href{http://arxiv.org/abs/2309.02533}{{\ttfamily arXiv:2309.02533
  [hep-th]}}.

\bibitem{Haehl:2023xys}
F.~M. Haehl, W.~Reeves, and M.~Rozali, ``{Symmetries and spectral statistics in
  chaotic conformal field theories. Part II. Maass cusp forms and arithmetic
  chaos},'' \href{http://dx.doi.org/10.1007/JHEP12(2023)161}{{\em JHEP}
  {\bfseries 12} (2023) 161}, \href{http://arxiv.org/abs/2309.00611}{{\ttfamily
  arXiv:2309.00611 [hep-th]}}.

\bibitem{DiUbaldo:2023qli}
G.~Di~Ubaldo and E.~Perlmutter, ``{AdS$_3$/RMT$_2$ Duality},''
  \href{http://arxiv.org/abs/2307.03707}{{\ttfamily arXiv:2307.03707
  [hep-th]}}.

\bibitem{Collier:2021rsn}
S.~Collier and A.~Maloney, ``{Wormholes and spectral statistics in the Narain
  ensemble},'' \href{http://dx.doi.org/10.1007/JHEP03(2022)004}{{\em JHEP}
  {\bfseries 03} (2022) 004}, \href{http://arxiv.org/abs/2106.12760}{{\ttfamily
  arXiv:2106.12760 [hep-th]}}.

\bibitem{Haehl:2023tkr}
F.~M. Haehl, C.~Marteau, W.~Reeves, and M.~Rozali, ``{Symmetries and spectral
  statistics in chaotic conformal field theories},''
  \href{http://arxiv.org/abs/2302.14482}{{\ttfamily arXiv:2302.14482
  [hep-th]}}.

\bibitem{Benjamin:2021ygh}
N.~Benjamin, S.~Collier, A.~L. Fitzpatrick, A.~Maloney, and E.~Perlmutter,
  ``{Harmonic analysis of 2d CFT partition functions},''
  \href{http://dx.doi.org/10.1007/JHEP09(2021)174}{{\em JHEP} {\bfseries 09}
  (2021) 174}, \href{http://arxiv.org/abs/2107.10744}{{\ttfamily
  arXiv:2107.10744 [hep-th]}}.

\bibitem{Kusuki:2018wpa}
Y.~Kusuki, ``{Light Cone Bootstrap in General 2D CFTs and Entanglement from
  Light Cone Singularity},''
  \href{http://dx.doi.org/10.1007/JHEP01(2019)025}{{\em JHEP} {\bfseries 01}
  (2019) 025}, \href{http://arxiv.org/abs/1810.01335}{{\ttfamily
  arXiv:1810.01335 [hep-th]}}.

\bibitem{Collier:2018exn}
S.~Collier, Y.~Gobeil, H.~Maxfield, and E.~Perlmutter, ``{Quantum Regge
  Trajectories and the Virasoro Analytic Bootstrap},''
  \href{http://dx.doi.org/10.1007/JHEP05(2019)212}{{\em JHEP} {\bfseries 05}
  (2019) 212}, \href{http://arxiv.org/abs/1811.05710}{{\ttfamily
  arXiv:1811.05710 [hep-th]}}.

\bibitem{hormander2015analysis}
L.~H{\"o}rmander, {\em The Analysis of Linear Partial Differential Operators I:
  Distribution Theory and Fourier Analysis}.
\newblock Classics in Mathematics. Springer Berlin Heidelberg, 2015.
\newblock \url{https://books.google.com/books?id=aaLrCAAAQBAJ}.

\bibitem{Afkhami-Jeddi:2019zci}
N.~Afkhami-Jeddi, T.~Hartman, and A.~Tajdini, ``{Fast Conformal Bootstrap and
  Constraints on 3d Gravity},''
  \href{http://dx.doi.org/10.1007/JHEP05(2019)087}{{\em JHEP} {\bfseries 05}
  (2019) 087}, \href{http://arxiv.org/abs/1903.06272}{{\ttfamily
  arXiv:1903.06272 [hep-th]}}.

\bibitem{Benjamin:2022pnx}
N.~Benjamin and C.-H. Chang, ``{Scalar modular bootstrap and zeros of the
  Riemann zeta function},''
  \href{http://dx.doi.org/10.1007/JHEP11(2022)143}{{\em JHEP} {\bfseries 11}
  (2022) 143}, \href{http://arxiv.org/abs/2208.02259}{{\ttfamily
  arXiv:2208.02259 [hep-th]}}.

\bibitem{rademacher1938fourier}
H.~Rademacher and H.~S. Zuckerman, ``On the fourier coefficients of certain
  modular forms of positive dimension,'' {\em Annals of Mathematics} {\bfseries
  39} no.~2, (1938) 433--462.

\bibitem{rademacher1943expansion}
H.~Rademacher, ``On the expansion of the partition function in a series,'' {\em
  Annals of Mathematics} {\bfseries 44} no.~3, (1943) 416--422.

\bibitem{Eberhardt:2023xck}
L.~Eberhardt and S.~Mizera, ``{Evaluating one-loop string amplitudes},''
  \href{http://dx.doi.org/10.21468/SciPostPhys.15.3.119}{{\em SciPost Phys.}
  {\bfseries 15} no.~3, (2023) 119},
  \href{http://arxiv.org/abs/2302.12733}{{\ttfamily arXiv:2302.12733
  [hep-th]}}.

\bibitem{Baccianti:2025gll}
M.~M. Baccianti, J.~Chandra, L.~Eberhardt, T.~Hartman, and S.~Mizera,
  ``{Rademacher expansion of modular integrals},''
  \href{http://arxiv.org/abs/2501.13827}{{\ttfamily arXiv:2501.13827
  [hep-th]}}.

\bibitem{Pal:2020wwd}
S.~Pal and Z.~Sun, ``{High Energy Modular Bootstrap, Global Symmetries and
  Defects},'' \href{http://dx.doi.org/10.1007/JHEP08(2020)064}{{\em JHEP}
  {\bfseries 08} (2020) 064}, \href{http://arxiv.org/abs/2004.12557}{{\ttfamily
  arXiv:2004.12557 [hep-th]}}.

\bibitem{Harlow:2021trr}
D.~Harlow and H.~Ooguri, ``{A universal formula for the density of states in
  theories with finite-group symmetry},''
  \href{http://dx.doi.org/10.1088/1361-6382/ac5db2}{{\em Class. Quant. Grav.}
  {\bfseries 39} no.~13, (2022) 134003},
  \href{http://arxiv.org/abs/2109.03838}{{\ttfamily arXiv:2109.03838
  [hep-th]}}.

\bibitem{Kang:2022orq}
M.~J. Kang, J.~Lee, and H.~Ooguri, ``{Universal formula for the density of
  states with continuous symmetry},''
  \href{http://dx.doi.org/10.1103/PhysRevD.107.026021}{{\em Phys. Rev. D}
  {\bfseries 107} no.~2, (2023) 026021},
  \href{http://arxiv.org/abs/2206.14814}{{\ttfamily arXiv:2206.14814
  [hep-th]}}.

\bibitem{Lin:2022dhv}
Y.-H. Lin, M.~Okada, S.~Seifnashri, and Y.~Tachikawa, ``{Asymptotic density of
  states in 2d CFTs with non-invertible symmetries},''
  \href{http://dx.doi.org/10.1007/JHEP03(2023)094}{{\em JHEP} {\bfseries 03}
  (2023) 094}, \href{http://arxiv.org/abs/2208.05495}{{\ttfamily
  arXiv:2208.05495 [hep-th]}}.

\bibitem{Lin:2021udi}
Y.-H. Lin and S.-H. Shao, ``{$\mathbb{Z}_N$ symmetries, anomalies, and the
  modular bootstrap},''
  \href{http://dx.doi.org/10.1103/PhysRevD.103.125001}{{\em Phys. Rev. D}
  {\bfseries 103} no.~12, (2021) 125001},
  \href{http://arxiv.org/abs/2101.08343}{{\ttfamily arXiv:2101.08343
  [hep-th]}}.

\bibitem{Lin:2023uvm}
Y.-H. Lin and S.-H. Shao, ``{Bootstrapping Non-invertible Symmetries},''
  \href{http://arxiv.org/abs/2302.13900}{{\ttfamily arXiv:2302.13900
  [hep-th]}}.

\bibitem{Dey:2024nje}
I.~Dey, S.~Pal, and J.~Qiao, ``{A universal inequality on the unitary 2D CFT
  partition function},'' \href{http://arxiv.org/abs/2410.18174}{{\ttfamily
  arXiv:2410.18174 [hep-th]}}.

\bibitem{Strominger:1996sh}
A.~Strominger and C.~Vafa, ``{Microscopic origin of the Bekenstein-Hawking
  entropy},'' \href{http://dx.doi.org/10.1016/0370-2693(96)00345-0}{{\em Phys.
  Lett. B} {\bfseries 379} (1996) 99--104},
  \href{http://arxiv.org/abs/hep-th/9601029}{{\ttfamily arXiv:hep-th/9601029}}.

\bibitem{Strominger:1997eq}
A.~Strominger, ``{Black hole entropy from near horizon microstates},''
  \href{http://dx.doi.org/10.1088/1126-6708/1998/02/009}{{\em JHEP} {\bfseries
  02} (1998) 009}, \href{http://arxiv.org/abs/hep-th/9712251}{{\ttfamily
  arXiv:hep-th/9712251}}.

\bibitem{Sen:2007qy}
A.~Sen, ``{Black Hole Entropy Function, Attractors and Precision Counting of
  Microstates},'' \href{http://dx.doi.org/10.1007/s10714-008-0626-4}{{\em Gen.
  Rel. Grav.} {\bfseries 40} (2008) 2249--2431},
  \href{http://arxiv.org/abs/0708.1270}{{\ttfamily arXiv:0708.1270 [hep-th]}}.

\bibitem{Sen:2012dw}
A.~Sen, ``{Logarithmic Corrections to Schwarzschild and Other Non-extremal
  Black Hole Entropy in Different Dimensions},''
  \href{http://dx.doi.org/10.1007/JHEP04(2013)156}{{\em JHEP} {\bfseries 04}
  (2013) 156}, \href{http://arxiv.org/abs/1205.0971}{{\ttfamily arXiv:1205.0971
  [hep-th]}}.

\bibitem{Verlinde:2000wg}
E.~P. Verlinde, ``{On the holographic principle in a radiation dominated
  universe},'' \href{http://arxiv.org/abs/hep-th/0008140}{{\ttfamily
  arXiv:hep-th/0008140}}.

\bibitem{Kutasov:2000td}
D.~Kutasov and F.~Larsen, ``{Partition sums and entropy bounds in weakly
  coupled CFT},'' \href{http://dx.doi.org/10.1088/1126-6708/2001/01/001}{{\em
  JHEP} {\bfseries 01} (2001) 001},
  \href{http://arxiv.org/abs/hep-th/0009244}{{\ttfamily arXiv:hep-th/0009244}}.

\bibitem{Shaghoulian:2015kta}
E.~Shaghoulian, ``{Modular forms and a generalized Cardy formula in higher
  dimensions},'' \href{http://dx.doi.org/10.1103/PhysRevD.93.126005}{{\em Phys.
  Rev. D} {\bfseries 93} no.~12, (2016) 126005},
  \href{http://arxiv.org/abs/1508.02728}{{\ttfamily arXiv:1508.02728
  [hep-th]}}.

\bibitem{Cardy:2017qhl}
J.~Cardy, A.~Maloney, and H.~Maxfield, ``{A new handle on three-point
  coefficients: OPE asymptotics from genus two modular invariance},''
  \href{http://dx.doi.org/10.1007/JHEP10(2017)136}{{\em JHEP} {\bfseries 10}
  (2017) 136}, \href{http://arxiv.org/abs/1705.05855}{{\ttfamily
  arXiv:1705.05855 [hep-th]}}.

\bibitem{Belin:2021ryy}
A.~Belin, J.~de~Boer, and D.~Liska, ``{Non-Gaussianities in the statistical
  distribution of heavy OPE coefficients and wormholes},''
  \href{http://dx.doi.org/10.1007/JHEP06(2022)116}{{\em JHEP} {\bfseries 06}
  (2022) 116}, \href{http://arxiv.org/abs/2110.14649}{{\ttfamily
  arXiv:2110.14649 [hep-th]}}.

\bibitem{Anous:2021caj}
T.~Anous, A.~Belin, J.~de~Boer, and D.~Liska, ``{OPE statistics from
  higher-point crossing},''
  \href{http://dx.doi.org/10.1007/JHEP06(2022)102}{{\em JHEP} {\bfseries 06}
  (2022) 102}, \href{http://arxiv.org/abs/2112.09143}{{\ttfamily
  arXiv:2112.09143 [hep-th]}}.

\bibitem{Benjamin:2023qsc}
N.~Benjamin, J.~Lee, H.~Ooguri, and D.~Simmons-Duffin, ``{Universal Asymptotics
  for High Energy CFT Data},''
  \href{http://arxiv.org/abs/2306.08031}{{\ttfamily arXiv:2306.08031
  [hep-th]}}.

\bibitem{Benjamin:2024kdg}
N.~Benjamin, J.~Lee, S.~Pal, D.~Simmons-Duffin, and Y.~Xu, ``{Angular fractals
  in thermal QFT},'' \href{http://dx.doi.org/10.1007/JHEP11(2024)134}{{\em
  JHEP} {\bfseries 11} (2024) 134},
  \href{http://arxiv.org/abs/2405.17562}{{\ttfamily arXiv:2405.17562
  [hep-th]}}.

\end{thebibliography}\endgroup
\bibliographystyle{utphys}

\end{document}